\documentclass[11pt]{article}

\usepackage[margin=1in]{geometry}
\usepackage{setspace}
\usepackage{amsmath, amssymb, amsthm, mathtools}
\usepackage{bm}
\usepackage{authblk}
\usepackage{orcidlink}

\usepackage{algorithm}
\usepackage{algorithmic}
\usepackage{booktabs}

\usepackage{natbib}
\usepackage{hyperref}
\usepackage{cleveref}
\hypersetup{
	colorlinks=true,
	linkcolor=blue,
	citecolor=blue,
	urlcolor=blue
}

\theoremstyle{plain}
\newtheorem{theorem}{Theorem}
\newtheorem{proposition}{Proposition}
\newtheorem{lemma}{Lemma}
\newtheorem{corollary}{Corollary}
\newtheorem{assumption}{Assumption}

\theoremstyle{definition}

\theoremstyle{remark}
\newtheorem{remark}{Remark}

\def\E{\mathbb{E}}
\def\R{\mathbb{R}}

\def\X{\bm{X}}
\def\u{\bm{u}}
\def\v{\bm{v}}

\newcommand{\identity}{\mathbf{I}}
\newcommand{\Prob}{\mathbb{P}}

\newcommand{\boldalpha}{\bm{\alpha}}
\newcommand{\boldtheta}{\bm{\theta}}
\newcommand{\boldmu}{\bm{\mu}}

\newcommand{\abs}[1]{\left\lvert #1 \right\rvert}
\newcommand{\norm}[1]{\left\| #1 \right\|}

\DeclareMathOperator*{\argmax}{arg\,max}

\title{Projection Diagnostics for Directional Asymmetry and Tail-Ratio Departure in Multivariate Data}

\author[1]{Sayantan Banerjee \orcidlink{0000-0001-5414-4817}}
\affil[1]{OM \& QT Area\\Indian Institute of Management Indore
}

\author[2]{Soudeep Deb \orcidlink{0000-0003-0567-7339}}
\affil[2]{Decision Sciences Area \\ Indian Institute of Management Bangalore}

\date{}

\begin{document}
	
    \maketitle
	
	\begin{abstract}
    We study projection-based diagnostics for distinguishing directional asymmetry from tail-ratio departure in multivariate data. The procedure reduces the problem to one-dimensional projections and computes two quantile-based summaries: a directional skewness measure evaluated over several quantile levels, and an interquantile tail-ratio evaluated relative to a chosen benchmark. The two summaries lead to a four-regime classification: symmetric benchmark-tail, symmetric tail-departed, skewed benchmark-tail, and skewed tail-departed. The quantile formulation avoids relying on third and fourth moments, which can be unstable in heavy-tailed settings. We establish population properties under central symmetry and ellipticity, uniform finite-sample bounds over the searched directions, and consistency of the threshold classifier under separated regimes. A sparse rank-one calculation is also used to show why coordinate directions can complement random directions in high dimensions. The resulting diagnostic is meant to guide subsequent modelling choices, for example whether a symmetric, skewed, tail-departed, or combined multivariate model is appropriate.

    \vspace{0.1in}

	\end{abstract}

\section{Introduction}\label{sec:introduction}

Projection-based approaches to multivariate non-normality have a long history. \citet{malkovich1973tests} proposed multivariate normality tests by maximizing univariate skewness, kurtosis, and Shapiro--Wilk type statistics over projections. \citet{baringhaus1991limit} later derived limit distributions for projection-based measures of multivariate skewness and kurtosis under elliptical distributions. An important lesson from this literature is that projection skewness alone is not a test of multivariate normality. In particular, Baringhaus and Henze showed that tests based on projection skewness in the Malkovich--Afifi sense can be inconsistent against fixed elliptically symmetric non-normal alternatives. This observation directly motivates our separation of the diagnostic into two components: directional skewness is used to detect violations of central symmetry, while the tail-ratio component is used to detect heavier- or lighter-tailed behavior relative to a chosen benchmark. Related projection-pursuit ideas also appear in the broader literature on exploratory multivariate analysis and non-Gaussianity. In this paper, we build a robust diagnostic framework based on directional quantiles, using separate projection summaries for asymmetry and tail-ratio departure and combining them into an interpretable four-regime classification.

The quantile-based aspect connects the paper to directional quantile methods, including the quantile tomography framework of \citet{kong2012quantile}, where multivariate information is studied through quantiles of one-dimensional projections. Our tail functional is related to interquantile tail-weight measures studied as robust alternatives to moment kurtosis \citep{ruppert1987what,jones2011skewness}; related quantile spacings also appear in earlier robust location work \citep{crow1967robust}. The present paper uses these ideas directionally: the tail ratio is evaluated across projections and paired with a directional skewness functional, allowing skewness and tail thickness to be separated rather than merged into a single omnibus statistic.

Classical tests of multivariate normality provide important benchmarks, but they are not designed to decompose the source of non-Gaussianity. Mardia's multivariate skewness and kurtosis measures are perhaps the most widely used moment-based summaries, and their asymptotic null distributions yield separate tests for skewness and kurtosis \citep{mardia1970measures, mardia1974applications}.  However, because they depend on third and fourth moments, their behavior may be sensitive to heavy tails and to high-dimensional finite-sample effects.  Omnibus procedures such as the Henze--Zirkler test \citep{henze1990class}, the Doornik--Hansen test \citep{doornik2008omnibus}, and energy-distance tests of normality \citep{szekely2013energy, szekely2017energy} are powerful tools for detecting general departures from multivariate normality, but an omnibus rejection does not identify whether the departure is driven by skewness, tail inflation, contamination, or a combination of these features. As mentioned above, projection-based competitors are closer in spirit to our approach, but they remain moment-based and are usually framed as tests of normality rather than as separate diagnostics for asymmetry and tail-ratio departure. Our numerical experiments further reveal that such tests are extremely time-consuming and therefore an alternative is necessary. Recent work on sub-dimensional Mardia measures meanwhile emphasizes that skewness and kurtosis may be localized in lower-dimensional features of a multivariate distribution, motivating procedures that search over directions or subspaces rather than relying only on global summaries \citep{chowdhury2022sub}.  Our contribution is to combine this projection viewpoint with robust directional quantiles, yielding separate interpretable statistics for skewness and tail behavior.

Our aim is to construct a projection-based diagnostic that separates directional skewness from directional tail-ratio departure in multivariate data. Unlike omnibus tests of multivariate normality, the proposed procedure is meant to indicate the type of departure, rather than only reject a Gaussian null. Let \(\bm X_1,\ldots,\bm X_n\) be independent \(p\)-dimensional observations, and let \(\mathcal{U}_{m}=\{\bm u_{1},\hdots, \bm u_{m}\}\) be a collection of directions sampled uniformly from the unit sphere $\mathbb{S}^{p-1}$. For each \(\bm u\in\mathcal U_m\), the projected sample \(\bm u^\top\bm X_1,\ldots,\bm u^\top\bm X_n\) is treated as a univariate sample. We compute a quantile-based directional skewness measure over a finite set of skewness levels \(\mathcal A_S\), and an interquantile tail-ratio relative to a chosen benchmark. Maximizing these quantities over the searched directions gives two diagnostic summaries, \(\widehat S_{\mathcal U_m,\mathcal A_S}\) and \(\widehat T_{\mathcal U_m}\). The two summaries are interpreted jointly: small values of both suggest behavior close to the chosen benchmark; large \(\widehat T_{\mathcal U_m}\) with small \(\widehat S_{\mathcal U_m,\mathcal A_S}\) suggests symmetric tail-ratio departure; large \(\widehat S_{\mathcal U_m,\mathcal A_S}\) with moderate \(\widehat T_{\mathcal U_m}\) suggests directional asymmetry with limited tail-ratio departure; and large values of both suggest combined asymmetry and tail-ratio departure. The diagnostic is intended as a screening and exploratory tool before fitting a more specific symmetric, skewed, tail-departed, or combined multivariate model.

The construction has three useful features. First, it reduces a high-dimensional diagnostic question to a collection of one-dimensional calculations, so that quantile-based measures of skewness and tail-ratio departure can be used without estimating high-order multivariate moments. Second, the directional statistics are computed separately across projections, which makes the procedure straightforward to implement and easy to parallelize. Third, the two summaries retain separate information about asymmetry and tail-ratio departure, helping to decide whether a symmetric, skewed, tail-departed, or combined model is more appropriate. In the implementation used in this paper, the direction set combines random and coordinate directions. Random directions provide broad coverage of the sphere, while coordinate directions improve sensitivity to marginal or sparse departures. Although we focus on random and coordinate directions in the present numerical study, the same projection framework can accommodate other structured directions, such as directions obtained from principal component analysis (PCA) or independent component analysis (ICA), when these are relevant to the application.

% The main advantages of our proposed method are threefold. First, it reduces a high-dimensional diagnostic question to a collection of one-dimensional problems, allowing robust univariate measures of skewness and tail thickness to be used without estimating high-order multivariate moments. Second, it is computationally simple and parallelizable, since the directional statistics are calculated independently across projections.  Third, it gives a separate assessment of asymmetry and tail-ratio departure, thereby helping the analyst decide whether the additional complexity of a skewed heavy-tailed model is warranted. A further contribution of the paper is to connect these diagnostics statistics for multivariate skewness and tail-ratio departures with familiar dimension-reduction tools.  Particularly, we show how the proposed random-projection statistics can be compared with the directions obtained through principal component analysis (PCA), independent component analysis (ICA), or the low-dimensional visual summaries obtained through the stochastic neighbor embedding mechanisms (SNE or $t$-SNE).  These connections clarify when classical dimension-reduction methods reveal distributional departures and when additional random directional screening is needed.  The resulting framework is therefore useful both as a formal diagnostic for skewed and heavy-tailed modeling and as a distributional interpretation layer for exploratory multivariate analysis.

The organization of the paper is as follows. \Cref{sec:framework} introduces the projection-based skewness and tail-ratio diagnostics, including the moment-based and quantile-based versions, the aggregation over projection directions, and the calibration scheme. \Cref{sec:theory} studies the corresponding population and finite-sample properties. We first record the behavior under central symmetry and ellipticity, then give a random-direction detectability result, uniform empirical control for the quantile diagnostics, and consistency of the four-regime classifier under separated regimes. The section also includes a sparse rank-one calculation to explain why coordinate directions can complement random directions in high dimensions. \Cref{sec:simulation} presents Monte Carlo experiments evaluating the diagnostic across benchmark-tail, tail-departed, skewed, and combined regimes. \Cref{sec:real_data} applies the method to Indian air-quality data. We conclude the paper with some important remarks about limitations and future scopes in \Cref{sec:conclusions}.

% The organization of this paper is as follows. In \Cref{sec:framework}, we first define a class of projection-based skewness and tail diagnostics and discuss choices of directional indices, including moment-based and quantile-based alternatives. We then establish basic population properties under elliptical symmetry and under directional skew alternatives, in \Cref{sec:theory}.  In particular, we show that the population directional skewness vanishes uniformly under centrally symmetric elliptical models, whereas it is nonzero whenever the distribution possesses a detectable asymmetric projection. A detailed empirical analysis through suitable Monte Carlo simulations is presented in \Cref{sec:simulation}. This helps us in evaluating the ability of the proposed diagnostics to distinguish Gaussian, symmetric heavy-tailed, skewed light-tailed and skewed heavy-tailed regimes. A real-data application with Indian air-quality data is presented in \Cref{sec:real_data}. We conclude the paper with some important remarks about limitations and future scopes in \Cref{sec:conclusions}.

\section{Mathematical framework}
\label{sec:framework}

In order to formally describe the procedure, let $\X$ be a $p$-dimensional random vector (where $p \geqslant 2$), with distribution function denoted by $P_{\X}$. For vectors $\bm a, \bm b \in \R^{p}$, let $\bm a^{\top} \bm b$ denote the Euclidean inner product and let $\norm{\bm a}=\sqrt{\bm a^{\top} \bm a}$ denote the Euclidean norm. The unit sphere in $\R^{p}$ is defined as \(\mathbb{S}^{p-1} = \left\{\bm u \in \R^{p} : \norm{\bm u} = 1 \right\}.\)

Let $\X_1,\ldots,\X_n$ be independent observations in $\R^p$, drawn from the distribution $P_{\X}$. As mentioned earlier, based on this sample, we are interested in diagnosing two different types of departures from a Gaussian or, more generally, an elliptically symmetric working model: the first is lack of central symmetry, and the second is tail-ratio behavior that differs from a chosen Gaussian or elliptical benchmark. These two features are often confounded in omnibus tests of multivariate normality. Prior to fitting a fully parametric skewed or heavy-tailed family, we aim to construct low-dimensional summaries which can indicate whether such a model is warranted. The basic device in this regard is projection. For each $\u \in \mathbb{S}^{p-1}$, the one-dimensional projection of $\X_i$ (for $i=1,\ldots,n$) in the direction $\u$ is given by
\begin{equation*}
%    \label{eq:projection-u}
    Y_i(\u) = \u^\top \X_i.
\end{equation*}

The collection \(Y_1(\u),\ldots,Y_n(\u)\) is a univariate sample obtained by viewing the data along the direction $\u$. If the multivariate distribution is centrally symmetric around a point $\boldtheta$, then \(\u^\top(\X-\boldtheta)\) is symmetric about zero for every $\u$. Similarly, if the distribution is elliptical, the standardized one-dimensional projections have a common shape, up to scale. Thus, directional projections give a natural way of examining both asymmetry and tail-ratio departure without estimating high-order multivariate moments directly. 

In practice though, the center of the distribution is unknown. We therefore work with centered observations $\widetilde{\X}_i=\X_i-\widehat{\boldtheta}$, where \(\widehat{\boldtheta}\) is a suitable location estimate. The ordinary sample mean may be used when the data appear reasonably regular. For heavy-tailed data, a coordinate-wise median, geometric median \citep[see][]{chaudhuri1996geometric}, or another robust location estimator is preferable. Unless otherwise stated, all projections below are computed from \(\widetilde{\X}_i\), and we write
\(
    \widetilde Y_i(\u) = \u^\top \widetilde{\X}_i.
\)
For notational simplicity, hereafter we suppress the tilde in some formulas when no confusion is likely.

\subsection{Directional skewness functionals}
\label{subsec:directional-skewness}

For each direction $\u$, let \(F_{\u}\) denote the distribution function of \(\u^\top(\X-\boldtheta)\). A directional skewness functional is a map
\[
\gamma:\mathbb{S}^{p-1}\to\mathbb R, \qquad u\mapsto \gamma(\u),
\]
where \(\gamma(\u)\) measures the asymmetry of the one-dimensional distribution \(F_{\u}\). The only essential requirement is that the functional vanishes under symmetry. More precisely, if \(F_{\u}\) is symmetric about zero, then \(\gamma(\u)=0\). Different choices of \(\gamma\) lead to different versions of the diagnostic. For example, the classical moment-based directional skewness is
\begin{equation*}
    \gamma_{M}(\u) = \frac{\E\{\u^\top(\X-\boldtheta)\}^3}
    {\left[\E\{\u^\top(\X-\boldtheta)\}^2\right]^{3/2}},
\end{equation*}
provided the required moments exist and the denominator is nonzero. Its sample version is naturally defined as
\begin{equation}
\label{eq:moment_based_skewness_statistic}
    \widehat\gamma_{M}(\u) = \frac{n^{-1}\sum_{i=1}^n \{\u^\top(\X_i-\widehat{\boldtheta})\}^3}
    {\left[n^{-1}\sum_{i=1}^n \{\u^\top(\X_i-\widehat{\boldtheta})\}^2 \right]^{3/2}}.
\end{equation}

Although the above is familiar and easy to interpret, it is also sensitive to a few extreme observations. Since one of the purposes of our paper is to separate skewness from heavy tails, the moment-based version should not be the only one considered. We therefore turn to a more robust alternative that is based on projected quantiles. For \(0<a<1\), define
\[
Q_a(\u)=\inf\{y:F_{\u}(y) \geqslant a\},
\]
and let \(\widehat Q_a(\u)\) be the corresponding empirical quantile of $\u^\top(X_1-\widehat{\boldtheta}),\ldots,\u^\top(X_n-\widehat{\boldtheta})$. The Bowley-type directional skewness functional is
\[
\gamma(\u)
=
\frac{
	Q_{0.75}(\u)+Q_{0.25}(\u)-2Q_{0.50}(\u)
}{
	Q_{0.75}(\u)-Q_{0.25}(\u)
},
\]
whenever the denominator is positive. This measure compares the upper and lower halves of the projected distribution around the median. If \(F_{\u}\) is symmetric, then $Q_{0.75}(\u)+Q_{0.25}(\u)-2Q_{0.50}(\u)=0$, so that \(\gamma(\u)=0\). Unlike the third-moment skewness, the quantile-based version remains meaningful for distributions whose third moment does not exist. For \(a\in (0,1/2)\), define the directional quantile skewness
\[
\gamma_{a}(\u)
=
\frac{Q_{1-a}(\u)+Q_a(\u)-2Q_{0.50}(\u)}
     {Q_{1-a}(\u)-Q_a(\u)},
\]
whenever \(Q_{1-a}(u)>Q_a(u)\). This statistic measures the deviation from symmetry by focusing on the $a^{th}$ quantiles on both tails, and it takes the form of the Bowley-type statistic if we set $a=0.25$. Smaller values of \(a\) use more extreme quantiles and can be more sensitive to tail asymmetry. The general empirical version is
\begin{equation}
\label{eq:quantile_based_skewness_statistic}
    \widehat\gamma_{a}(\u) = \frac{\widehat Q_{1-a}(\u)+\widehat Q_{a}(\u)-2\widehat Q_{0.50}(\u)}
    {\widehat Q_{1-a}(\u)-\widehat Q_{a}(\u)},
\end{equation}
for quantile level $a \in (0,0.5)$.  Of course, other robust skewness functionals, such as the medcouple \citep{brys2004robust}, may also be used inside the same framework. We focus on the moment and quantile versions because they have complementary roles: the former is classical and powerful under finite moments, while the latter is stable under heavy tails and local contamination.

%\sd{throughout, I think it would be better not to keep the subscript $Q$ for the quantile based definitions, only the small $q$ or small $a$ to reflect which quantile should be good enough. I am trying to make sure all definitions and notations follow this, please double check the same}

\subsection{Directional tail-weight functionals}
\label{subsec:directional-tail}

A directional tail-weight functional is a map
\[
\tau:\mathbb{S}^{p-1}\to\mathbb R, \qquad u\mapsto \tau(\u)
\]
where \(\tau(\u)\) measures the tail-ratio departure of the projection
\(u^\top(\X-\boldtheta)\). As before, several choices are possible. For instance, we may consider the classical moment-based choice of projected kurtosis:
\[
\tau_{M}(\u)
=
\frac{
	E\{\u^\top(\X-\boldtheta)\}^4
}{
	\left[E\{\u^\top(\X-\boldtheta)\}^2\right]^2
},
\]
whenever the fourth moment exists. The sample version of this is
\begin{equation}
\label{eq:moment_based_tail_statistic}
    \widehat\tau_{M}(\u) = \frac{n^{-1}\sum_{i=1}^n \{\u^\top(\X_i-\widehat{\boldtheta})\}^4}
    {\left[n^{-1}\sum_{i=1}^n \{\u^\top(\X_i-\widehat{\boldtheta})\}^2\right]^2}.
\end{equation}
For a Gaussian projection, the population value is \(3\). Thus one may also work with excess projected kurtosis, \(\tau_{M}(\u)-3.\) We highlight a clear drawback in this case; namely, in very heavy-tailed settings, the fourth moment may not exist, and even when it does exist the empirical kurtosis may have high variance. For this reason we also use a quantile tail-ratio. Fix a tail probability level \(q\in(0,1/2)\), for example \(q=0.025\) or \(q=0.05\). Then, define
\[
\tau_q(\u)
=
\frac{
	Q_{1-q}(\u)-Q_q(\u)
}{
	Q_{0.75}(\u)-Q_{0.25}(\u)
},
\]
the empirical version of which is
\begin{equation}
\label{eq:quantile_based_tail_statistic}
    \widehat\tau_{q}(\u) = \frac{\widehat Q_{1-q}(\u)-\widehat Q_q(\u)}
    {\widehat Q_{0.75}(\u)-\widehat Q_{0.25}(\u)}.
\end{equation}
This compares the width of an outer central interval to the inter-quartile range. It is invariant under location and scale transformations of the projected data. For a standard normal distribution, this takes the simpler form $\Phi^{-1}(1-q)/\Phi^{-1}(0.75)$. For example, when \(q=0.025\), one may calculate its value as approximately 2.91. Heavy-tailed projected distributions tend to have larger values of \(\tau_q(\u)\).

In the rest of the paper, we use the notations $\widehat\gamma(\u)$ and \(\widehat\tau(\u)\) for generic directional skewness statistic and directional tail statistic. In applications and simulations, both the moment-based definitions (\eqref{eq:moment_based_skewness_statistic} and \eqref{eq:moment_based_tail_statistic}), or the quantile tail-ratio-based statistics (\eqref{eq:quantile_based_skewness_statistic} and \eqref{eq:quantile_based_tail_statistic}) can be examined. For formal high-dimensional statements, the quantile versions are often more convenient, because it does not require fourth moments.

\subsection{Aggregating directional evidence}
\label{subsec:aggregation}

Let $\mathcal U_m=\{\u_1,\ldots,\u_m\}\subset \mathbb{S}^{p-1}$ be a collection of directions in the unit sphere. The simplest choice is to draw the directions independently and uniformly from \(\mathbb{S}^{p-1}\). Let \(\mathcal A_S\subset(0,1/2)\) be a finite set of skewness levels. For \(a\in\mathcal A_S\), and for a direction set \(\mathcal U_m\), the corresponding population skewness summary taking \(\gamma_{a}(\u)\) as the skewness measure is given by,
\begin{align}
S_{\mathcal U_m,\mathcal A_S}
=
\max_{\u\in\mathcal U_m}
\max_{a\in\mathcal A_S}
\abs{\gamma_{a}(\u)}.
\label{eq:skewness_measure}
\end{align}
The corresponding tail-weight measure involving \(\tau_q(\u)\)  is 
\begin{align}
T_{\mathcal U_m} = \max_{\u\in\mathcal U_m} \abs{\tau_q(\u)-\tau_{q}^0}.
\label{eq:tail_measure}
\end{align}

Given \(\mathcal U_m\), define the empirical version of skewness and tail diagnostics as
\begin{equation}
    \widehat S_{\mathcal U_m,\mathcal A_S} = \max_{\u\in\mathcal U_m} \max_{a\in\mathcal A_S} \abs{\widehat\gamma_{a}(\u)}, \quad
    \widehat T_{\mathcal U_m} = \max_{\u\in\mathcal U_m} \abs{\widehat\tau_q(\u)-\tau_{q}^0}.
    \label{eq:empirical_skew_tail_measure}
\end{equation}
Here, \(\tau_{q}^0\) is a reference tail value. For the quantile tail-ratio, it can also be derived under Gaussian calibration as shown above. If the purpose is only to rank directions by tail-ratio departure, one may instead use $\widehat T_{\mathcal U_m} = \max_{\u\in\mathcal U_m} \widehat\tau_q(\u)$, but the centered version is more natural for testing against a specified reference law. Note that, we vary \(a\) in the skewness statistic because our simulations indicate that different skewness levels capture different forms of asymmetry, with the Bowley choice \(a=0.25\) sometimes weak for tail-driven skewness. The tail level \(q\) is kept fixed to maintain a stable and interpretable tail-ratio benchmark, since varying \(q\) gave little additional empirical gain while increasing sampling variability.

Apart from the maximum, for general functionals \(\gamma\) and \(\tau\), we may also consider upper quantiles or averages:
\begin{equation*}
    \widehat S_{\mathcal U_m,\mathcal A_S}^{(\alpha)} = \widehat Q_{1-\alpha} \left( \abs{\widehat\gamma(\u_1)}, \ldots ,\abs{\widehat\gamma(\u_m)} \right), \quad
    \widehat T_{\mathcal U_m}^{(\alpha)} = \widehat Q_{1-\alpha} \left(\abs{\widehat\tau(\u_1)-\tau_{0}},\ldots, \abs{\widehat\tau(\u_m)-\tau_{0}} \right),
\end{equation*}
for small \(\alpha\). These trimmed versions are less sensitive to a single unstable direction. Nevertheless, the maximum is useful for detecting localized departures and is the main object of the theoretical analysis.

The maximizing directions are also useful for interpretation. For skewness, one may record \((\widehat{\u}_S,\widehat a_S)\in\argmax_{\u\in\mathcal U_m,\ a\in\mathcal A_S}|\widehat\gamma_{a}(\u)|\), while for tail behavior one may record \(\widehat{\u}_T\in\argmax_{\u\in\mathcal U_m}|\widehat\tau_q(\u)-\tau_q^0|\). The direction \(\widehat{\u}_S\), together with the selected level \(\widehat a_S\), indicates whether the strongest asymmetry is central or tail-oriented, and \(\widehat{\u}_T\) identifies the linear combination with the largest tail-ratio departure from the reference. This makes the procedure more interpretable than a single omnibus rejection of multivariate normality.

We use a direction set that combines random and coordinate projections. The random directions \(\u_1,\ldots,\u_m\) are drawn independently from \(\mathrm{Unif}(\mathbb S^{p-1})\), for example by generating \(Z_j\sim N_p(0,I_p)\) and setting \(\u_j=Z_j/\|Z_j\|_2\). They give broad coverage of the sphere when the direction of non-Gaussianity is unknown.

In high dimensions, however, departures from Gaussianity may be localized in sparse directions, in which case purely random projections can have weak alignment with the signal. We augment the random directions with coordinate directions and take
\(\mathcal U_m = \mathcal U_{\mathrm{rand}} \cup \mathcal U_{\mathrm{coord}}\), where \(\mathcal U_{\mathrm{rand}}\) consists of $m$ independent random directions drawn uniformly from \(\mathbb S^{p-1}\), and \(\mathcal U_{\mathrm{coord}}=\{\bm{e}_1,\ldots,\bm{e}_p\}\) is the set of coordinate directions. The coordinate directions improve sensitivity to marginal asymmetry or marginal tail-ratio departures, while the random directions retain coverage of non-axis-aligned alternatives. This choice is computationally simple and matches the sparse-direction discussion in Section~\ref{subsec:sparse_directional_alternatives}.

Other structured directions can also be incorporated. For example, principal component directions may be useful when non-Gaussianity is associated with dominant variance modes, while independent component directions are natural when one wants to search for non-Gaussian linear combinations. Such random, optimized, and structured projections are classical in projection pursuit and related multivariate screening problems; see, for example, \citet{friedman1974projection,diaconis1984asymptotics,pena2001multivariate,hyvarinen2001ica,loperfido2018skewness}. We do not use PCA or ICA directions in the numerical implementation below; they are mentioned only as possible extensions of the same projection framework.

\subsection{Calibration}
\label{subsec:calibration}

The raw values of \(\widehat S_{\mathcal U_m,\mathcal A_S}\) and \(\widehat T_{\mathcal U_m}\) are not directly comparable across sample sizes, dimensions, or choices of \(m\). The maximum over many directions will naturally be larger than a statistic computed along a single direction. Calibration is therefore an essential part of the procedure. We propose to use a simulation-based calibration under a reference model. The simplest reference is the Gaussian model
\[
H_0^G : \; \X \sim N_p(\boldmu,\Sigma).
\]
Let \(\widehat{\boldmu}\) and \(\widehat\Sigma\) be estimates of location and scatter. Then, for bootstrap iterations \(b=1,\ldots,B\), one can generate a reference sample
\[
\X_1^{(b)},\ldots,\X_n^{(b)} \sim N_p(\widehat{\boldmu},\widehat\Sigma),
\]
and compute the corresponding reference statistics \(\widehat S_{\mathcal U_m,\mathcal A_S}^{(b)},\;\widehat T_{\mathcal U_m}^{(b)}\) using the same direction-generation scheme and the same directional functionals as for the observed data. Let \(\widehat S_{\mathcal U_m,\mathcal A_S}^{\mathrm{obs}}\) and \(\widehat T_{\mathcal U_m}^{\mathrm{obs}}\) denote the skewness and tail-ratio departure statistics for the observed sample. The calibrated \(p\)-values can then be calculated as
\begin{equation}
    p_S = \frac{1+\sum_{b=1}^B \mathbf 1\{\widehat S_{\mathcal U_m,\mathcal A_S}^{(b)}\geqslant \widehat S_{\mathcal U_m,\mathcal A_S}^{\mathrm{obs}}\}}
    {B+1}, \quad
    p_T = \frac{1+\sum_{b=1}^B\mathbf 1\{\widehat T_{\mathcal U_m}^{(b)}\geqslant \widehat T_{\mathcal U_m}^{\mathrm{obs}}\}}
    {B+1}.
    \label{eq:calibrated_p_val}
\end{equation}
Note that the addition of one in numerator and denominator avoids zero Monte Carlo \(p\)-values. 

When the dimension of the observations (\(p\)) is comparable to or larger than \(n\), the sample covariance matrix may be singular or poorly conditioned. In the Gaussian calibration we therefore use a regularized scatter estimate. A simple choice is a linear shrinkage estimator
\[
\widehat\Sigma_\rho=(1-\rho)\widehat\Sigma+\rho T,
\qquad 0\leqslant \rho \leqslant 1,
\]
where \(T\) may be \(\operatorname{diag}(\widehat\Sigma)\) or a multiple of the identity matrix. Shrinkage estimators of this type are standard in high-dimensional covariance estimation \citep{ledoit2004well,chen2010shrinkage}. When sparsity of the covariance matrix is plausible, thresholded covariance estimators provide another option \citep{bickel2008covariance,cai2011adaptive}. For heavy-tailed or elliptical calibration, robust scatter or shape estimators may be preferable, including classical \(M\)-estimators of scatter and Tyler's shape estimator \citep{maronna1976robust,tyler1987distribution}. The calibration step only requires a positive definite scatter or shape estimate; the same estimate is used throughout the bootstrap reference generation.

Gaussian calibration is useful when the working question is whether the data depart from Gaussian behavior through skewness or tail heaviness. In some applications, a broader reference class may be more appropriate. For example, one may wish to test directional asymmetry while allowing symmetric heavy tails. In that case, the skewness diagnostic can be calibrated under an elliptical reference model. While Gaussian calibration treats heavy tails as a departure, elliptical calibration uses them as part of the reference law when testing for asymmetry. For elliptical calibration, we propose to use the following semi-parametric bootstrap. Let \(\widehat{\boldtheta}\) be a robust location estimate, such as the spatial median, and let \(\widehat\Sigma\) be a positive definite scatter or shape estimate. In high-dimensional settings, \(\widehat\Sigma\) may be replaced by a regularized scatter estimate. Compute whitened residuals \(\bm{Z}_i=\widehat\Sigma^{-1/2}(\X_i-\widehat{\boldtheta})\) and radial components \(R_i=\norm{\bm{Z}_i}\). A bootstrap reference sample is generated by drawing \(R_i^\ast\) with replacement from \(R_1,\ldots,R_n\), drawing \(\u_i^\ast\) independently and uniformly from \(\mathbb{S}^{p-1}\), and setting
\[
\X_i^\ast=\widehat{\boldtheta}+\widehat\Sigma^{1/2}R_i^\ast \u_i^\ast .
\]
This preserves the empirical radial distribution while imposing spherical symmetry after whitening. We highlight here that the calibration is intended to be a diagnostic reference distribution only.

%\sd{this part maybe discussed directly in Sec 3.4?}

The calibrated values \(p_S\) and \(p_T\) give a two-dimensional diagnostic for classification into the following classes: symmetric benchmark-tail, symmetric tail-departed, skewed benchmark-tail, and skewed tail-departed. This classification should be used as a guide for subsequent modeling. For instance, a significant tail diagnostic with little evidence of skewness suggests that a symmetric heavy-tailed elliptical model may be adequate. Evidence in both directions points instead toward skewed heavy-tailed families, such as skew-\(t\) or related skew-elliptical models. We present this in more detail in Section~\ref{subsec:four_regime_consistency}.

% At a nominal level \(\alpha\), the rule is
% \[
% \widehat C =
% \begin{cases}
% 	\text{light-tailed symmetric}, 
% 	& p_S>\alpha,\ p_T>\alpha,\\[3pt]
% 	\text{symmetric heavy-tailed}, 
% 	& p_S>\alpha,\ p_T\leqslant \alpha,\\[3pt]
% 	\text{skewed light-tailed}, 
% 	& p_S\leqslant \alpha,\ p_T>\alpha,\\[3pt]
% 	\text{skewed heavy-tailed}, 
% 	& p_S\leqslant \alpha,\ p_T\leqslant \alpha.
% \end{cases}
% \]

\begin{remark}
	There might be cases where non-Gaussianity is concentrated in a few directions. The use of maxima makes the procedure sensitive to such localized departures. It also means that calibration must account for the number and type of directions searched over. We also note that the robust quantile-based version is often preferable in the present setting. Moment skewness and kurtosis are informative when moments exist and sample sizes are large, but they can be dominated by a small number of extreme observations. The quantile version sacrifices some efficiency under exact Gaussianity; however, it gives a more stable separation between asymmetry and tail-ratio departure in the heavy-tailed regimes that motivate this paper.
\end{remark}

\section{Theoretical results}
\label{sec:theory}

We next study the theoretical properties of the projection-based diagnostics introduced in Section~\ref{sec:framework}. First, we formalize the population-level separation between directional asymmetry and directional tail-ratio departure. We then explain when a finite collection of projections can detect a population departure. The results also provide finite-sample control for the empirical quantile-based diagnostics. Finally, we examine a stylized sparse alternative to clarify why random directions and coordinate directions can play complementary roles in high dimensions.

We focus on two scale-free population functionals discussed in the previous section. The first is the Bowley-type directional skewness \(\gamma_a(\u)\), and the second is the quantile tail-ratio \(\tau_q(\u)\). Both are invariant under location and positive scale transformations of the projected data. This invariance is useful because different directions may have different projected variances even under a purely elliptical model. 

\subsection{Population separation under symmetry and ellipticity}
\label{subsec:population-separation}

The first result records the fact that central symmetry of the multivariate law is inherited by every one-dimensional projection. Hence directional skewness must vanish simultaneously in all directions. This observation is consistent with the classical projection-based skewness literature; see, for example, \citet{malkovich1973tests,baringhaus1991limit}.

\begin{lemma}[Directional skewness under central symmetry]
	\label{lem:central_symmetry_skewness}
	Suppose \(\X-\boldtheta\stackrel{d}{=}-(\X-\boldtheta)\). Then, for every \(\u\in \mathbb{S}^{p-1}\), the projected variable \(Y(\u)=\u^\top(\X-\boldtheta)\) is symmetric about zero. If the projected quantiles \(Q_a(\u)\), \(Q_{0.50}(\u)\), and \(Q_{1-a}(\u)\) are unique and \(Q_{1-a}(\u)>Q_a(\u)\), then \(\gamma_{a}(\u)=0\) for every \(a\in\mathcal A_S\). If, in addition, \(E|Y(\u)|^3<\infty\), then the moment directional skewness also vanishes, namely \(E(Y(\u)^3)/\{E(Y(\u)^2)\}^{3/2}=0\), provided \(E(Y(\u)^2)>0\).
\end{lemma}

\begin{proof}
	Fix \(\u\in \mathbb{S}^{p-1}\). Since \(\X-\boldtheta\stackrel{d}{=}-(\X-\boldtheta)\), taking inner products with $\u$ gives \(Y(\u)\stackrel{d}{=}-Y(\u)\). Thus \(Y(\u)\) is symmetric about zero. By symmetry and uniqueness of the relevant projected quantiles, \(Q_{0.50}(\u)=0,\,Q_{1-a}(\u)=-Q_a(\u).\) Therefore \(Q_{1-a}(\u)+Q_a(\u)-2Q_{0.50}(\u)=0,\) and hence \(\gamma_{a}(\u)=0\).
    
	For the moment statement, \(Y(\u)\stackrel{d}{=}-Y(\u)\) implies \(E(Y(\u)^3)=E\{(-Y(\u))^3\}=-E(Y(\u)^3)\), whenever the third absolute moment is finite. Hence, \(E(Y(\u)^3)=0\), thereby completing the proof of the lemma.
\end{proof}

We next record the corresponding tail-weight benchmark under ellipticity. The fact that one-dimensional projections of an elliptical distribution differ only by scale is standard; see, for example,  \citet{cambanis1981elliptically} and \citet{fang1990symmetric}. 

Suppose \(\X\) admits the elliptical stochastic representation $\X=\boldtheta+RA\bm{U}$, where \(AA^\top=\Sigma\), \(R\geqslant 0\) is a scalar radial variable, and \(\bm{U}\) is uniformly distributed on \(\mathbb{S}^{p-1}\), independently of \(R\).  This representation is used with the convention that deterministic projection directions \(\u\in\mathbb{S}^{p-1}\) satisfying \(\u^\top\Sigma \u=0\) are degenerate and are excluded.

\begin{lemma}[Scale invariance of projected tail ratios under ellipticity]
	\label{lem:elliptical_tail_ratio}
	Suppose $\X=\boldtheta+RA\bm{U}$ is elliptically distributed as above. For every \(\u\in \mathbb{S}^{p-1}\) with \(\sigma_{\u}^2=\u^\top\Sigma \u>0\), there exists a scalar random variable \(Z_0\), not depending on $\u$, such that $\sigma_{\u} Z_0$ has the same distribution as \(\u^\top(\X-\boldtheta)\). Consequently, \(\tau_q(\u)\) is constant over all nondegenerate directions. That is, there exists a number \(\tau_q^{\rm ell}\), depending on the radial law and on \(q\), such that \(\tau_q(\u)=\tau_q^{\rm ell}\) whenever \(\u^\top\Sigma \u>0\).
\end{lemma}

\begin{proof}
	Fix a direction $\u$ with \(\sigma_{\u}^2=\u^\top\Sigma \u>0\), and put \(\bm{g}=A^\top \u\). Then \(\norm{\bm{g}}^2=\u^\top AA^\top \u=\u^\top\Sigma \u=\sigma_{\u}^2\). Since $\u$ is uniform on the sphere, rotational invariance gives \(\bm{g}^\top \bm{U}\stackrel{d}{=}\norm{\bm{g}} \bm{e}_1^\top \bm{U}=\sigma_{\u} e_1^\top \bm{U}\). Therefore
	\(\u^\top(\X-\boldtheta)=R \bm{g}^\top \bm{U}\stackrel{d}{=}\sigma_{\u} R(\bm{e}_1^\top \bm{U}).\) Let \(Z_0=R(\bm{e}_1^\top \bm{U})\), which proves the scale representation.
	
	Let \(Q_a^0\) be the \(a\)-quantile of \(Z_0\). Since \(\sigma_{\u}>0\), the quantiles of \(\sigma_{\u} Z_0\) are \(\sigma_{\u} Q_a^0\). Hence, the common factor \(\sigma_{\u}\) cancels in the ratio defining \(\tau_q(\u)\), giving
	\[
	\tau_q(\u) = \frac{Q_{1-q}^0-Q_q^0}{Q_{0.75}^0-Q_{0.25}^0},
	\]
	which does not depend on $\u$.
\end{proof}

\begin{corollary}[Gaussian and Student-\(t\) benchmarks]
	\label{cor:gaussian_t_benchmarks}
	Under \(N_p(\boldtheta,\Sigma)\), every nondegenerate standardized projection is standard normal, and therefore
	\[
	\tau_q^N = \frac{\Phi^{-1}(1-q)-\Phi^{-1}(q)} {\Phi^{-1}(0.75)-\Phi^{-1}(0.25)}.
	\]
	Under the usual multivariate Student-\(t_\nu\) model with location $\boldtheta$ and scatter \(\Sigma\), every nondegenerate standardized projection is univariate \(t_\nu\), and hence
	\[
	\tau_q^{t_\nu} = \frac{t_\nu^{-1}(1-q)-t_\nu^{-1}(q)}{t_\nu^{-1}(0.75)-t_\nu^{-1}(0.25)}.
	\]
\end{corollary}

\begin{proof}
	Both statements follow from the standard closure of Gaussian and elliptical Student-\(t\) distributions under linear projection. After standardization by the projected scale, the projected Gaussian law is \(N(0,1)\), while the projected multivariate \(t_\nu\) law is univariate \(t_\nu\). Substituting the corresponding quantiles into the definition of \(\tau_q\) gives the two displayed expressions.
\end{proof}

Lemma~\ref{lem:central_symmetry_skewness} and Lemma~\ref{lem:elliptical_tail_ratio} provide the population justification for using the pair of projection summaries. Under central symmetry, all directional skewness values vanish. Under an elliptical law, the scale-free tail-ratio does not vary with direction, although its common value may be Gaussian or heavier than Gaussian depending on the radial distribution. Thus, the skewness diagnostic targets violation of central symmetry, while the tail diagnostic measures departure from a chosen tail benchmark. The next subsection studies when a finite set of directions can find such a population departure.

\subsection{Random-direction detectability}
\label{subsec:random-direction-detectability}

The population results above describe what the directional functionals measure when all directions are available. In practice, however, the diagnostic searches only over a finite set of directions. We first consider the case where the directions are sampled independently and uniformly from the unit sphere. This gives a transparent bound on the probability that a finite random search finds a localized population departure. The use of random projections as a way to probe high-dimensional structure is closely related to the broader literature on low-dimensional projections of high-dimensional data. \citet{hall1993almost} showed that many low-dimensional projections of high-dimensional data can exhibit approximately simple behavior under suitable standardization. Our use of random directions is though different in purpose. We use them as a finite search device for localized departures in directional skewness and directional tail-ratio departure. The spherical-cap argument below quantifies the probability that such a search comes sufficiently close to a signal-bearing direction.

Let \(\psi:\mathbb{S}^{p-1}\to\mathbb R\) be a generic population directional functional. In the applications below, \(\psi(\u)\) is either the directional skewness \(\gamma_a(\u)\), or the centered tail-ratio \(\tau_q(\u)-\tau_q^0\), where \(\tau_q^0\) is a reference value such as the Gaussian benchmark \(\tau_q^N\). We use the geodesic distance \(d_{\mathbb S}(\u,\bm{v})=\arccos(\u^\top \bm{v})\) on \(\mathbb{S}^{p-1}\), and write \(B_{\mathbb S}(\u_0,r)=\{\u\in \mathbb{S}^{p-1}:d_{\mathbb S}(\u,\u_0)\leqslant r\}\) for the spherical cap of radius \(r\) around \(\u_0\). Let \(\zeta_p \) denote the uniform probability measure on \(\mathbb{S}^{p-1}\). 

The following mild regularity condition says only that the directional functional does not change abruptly near a signal-bearing direction.

\begin{assumption}[Local directional regularity]
	\label{ass:lipschitz_directional}
	There exists \(\u_0\in \mathbb{S}^{p-1}\), \(r_0\in(0,\pi]\), and \(L<\infty\) such that \(\abs{\psi(\u)-\psi(\bm v)} \leqslant L\,d_{\mathbb S}(\u,\bm v)\) for all \(\u,\bm v\in B_{\mathbb S}(\u_0,r_0)\). 
\end{assumption}

For quantile-based functionals, this condition follows if the relevant projected quantile curves are locally Lipschitz in the direction and the projected inter-quartile range is bounded away from zero. We state this reduction after the main detection result. Similar regularity assumptions are standard in quantile-process arguments \citep[see, for example,][Chapter~21]{vandervaart1998asymptotic}.

\begin{theorem}[Detection of a localized directional departure]
	\label{thm:random_direction_detection}
	Let \(\psi:\mathbb{S}^{p-1}\to\mathbb R\) be a population directional functional. Suppose that \(\abs{\psi(\u_0)} \geqslant \Delta\) for some \(\u_0\in \mathbb{S}^{p-1}\), \(\Delta>0\). Suppose also that Assumption~\ref{ass:lipschitz_directional} holds in a neighborhood of \(\u_0\). If \(L=0\), set \(r_\Delta=r_0\). Otherwise, set \(r_\Delta=\min\{r_0,\Delta/(2L)\}\). If \(\u_1,\ldots,\u_m\) are independently chosen uniform random directions on \(\mathbb{S}^{p-1}\), then
    \[
	\Prob\left\{\max_{1\leqslant j\leqslant m} \abs{\psi(\u_j)} \geqslant \Delta/2\right\}
	\geqslant
	1-\left[1-\zeta_p \left\{B_{\mathbb S}(\u_0,r_\Delta)\right\}\right]^m .
	\]
\end{theorem}

\begin{proof}
	Let \(A_\Delta=B_{\mathbb S}(\u_0,r_\Delta)\). For any \(\u\in A_\Delta\), Assumption~\ref{ass:lipschitz_directional} gives \(\abs{\psi(\u)-\psi(\u_0)} \leqslant Ld_{\mathbb S}(\u,\u_0) \leqslant Lr_\Delta \leqslant \Delta/2\). Hence, by the reverse triangle inequality, \(\abs{\psi(\u)} \geqslant \abs{\psi(\u_0)} - \abs{\psi(\u)-\psi(\u_0)} \geqslant \Delta/2\). Therefore, if at least one sampled direction belongs to \(A_\Delta\), then \(\max_{1\leqslant j\leqslant m}|\psi(\u_j)|\geqslant \Delta/2\). Since the directions are independent and uniform on \(\mathbb{S}^{p-1}\),
	\[
	\Prob\{\exists j:\u_j\in A_\Delta\}
	=
	1-\left\{1-\zeta_p (A_\Delta)\right\}^m.
	\]
	The claim follows.
\end{proof}

The theorem says that the departure need not be visible in many directions. It is enough that it persists in a spherical neighborhood of one direction. The probability of detection is controlled by the mass of that neighborhood and the number of random projections.

The spherical-cap probability has a standard expression. If $\u$ is uniform on \(\mathbb{S}^{p-1}\), then by rotational invariance, \(\zeta_p \{B_{\mathbb S}(\u_0,r)\}=\Prob(U_1\geqslant \cos r)\), where \(U_1\) is the first coordinate of $\u$. For \(p \geqslant 2\), \(U_1\) has density \(f_p(t)=c_p(1-t^2)^{(p-3)/2}\), \(-1<t<1\), with \(c_p=\Gamma(p/2)/\{\sqrt{\pi}\Gamma((p-1)/2)\}\). Hence
\[
\zeta_p \{B_{\mathbb S}(\u_0,r)\}
=
\int_{\cos r}^{1} c_p \left(1-t^2\right)^{(p-3)/2}\,dt .
\]
These facts are standard consequences of rotational invariance on the sphere and are closely related to concentration of measure on \(\mathbb{S}^{p-1}\); see, for example, \citet{ledoux2001concentration} and \citet{vershynin2018highdimensional}. Note that, for small \(r\), the cap probability can be very small, so a purely random search may require many directions when the signal region is narrow. Further, the theorem immediately implies the corresponding statements for the two diagnostics. 

\begin{corollary}[Random-direction detection of skewness]
	\label{cor:random_direction_skewness}
	Suppose that there exist \(\u_0\in S^{p-1}\) and \(a_0\in\mathcal A_S\) such that \(|\gamma_{a_0}(\u_0)|\geqslant \Delta_S>0.\) Suppose also that \(\u\mapsto \gamma_{a_0}(\u)\) is locally \(L_S\)-Lipschitz in a neighborhood of \(\u_0\). If \(\u_1,\ldots,\u_m\) are independent uniform
directions on \(S^{p-1}\), then
\[
\Prob\left\{
\max_{1\leqslant j\leqslant m}
\max_{a\in\mathcal A_S}
|\gamma_{a}(\u_j)|
\geqslant \Delta_S/2
\right\}
\ge
1-\left[1-\zeta_p\{B_{\mathbb S}(\u_0,r_S)\}\right]^m,
\]
where \(r_S=\min\{r_0,\Delta_S/(2L_S)\}\), with the convention \(r_S=r_0\) when \(L_S=0\).
\end{corollary}

% \begin{proof}
% 	Apply Theorem~\ref{thm:random_direction_detection} with \(\psi(\u)=\gamma_Q(\u)\).
% \end{proof}

\begin{corollary}[Random-direction detection of tail-ratio departure]
	\label{cor:random_direction_tail}
	Let \(\tau_q^0\) be a reference tail value. Suppose that \(\abs{\tau_q(\u_0)-\tau_q^0} \geqslant \Delta_T>0\) for some \(\u_0\in \mathbb{S}^{p-1}\), and that \(\u\mapsto\tau_q(\u)-\tau_q^0\) is locally \(L_T\)-Lipschitz in a neighborhood of \(\u_0\). If \(\u_1,\ldots,\u_m\) are independent uniform directions on \(\mathbb{S}^{p-1}\), then for \(r_T=\min\{r_0,\Delta_T/(2L_T)\}\),
	\[
	\Prob\left\{
	\max_{1\leqslant j\leqslant m} \abs{\tau_q(\u_j)-\tau_q^0}
	\geqslant \Delta_T/2
	\right\}
	\geqslant
	1-\left[1-\zeta_p \left\{B_{\mathbb S}(\u_0,r_T)\right\}\right]^m.
	\]
\end{corollary}

% \begin{proof}
% 	Apply Theorem~\ref{thm:random_direction_detection} with \(\psi(\u)=\tau_q(\u)-\tau_q^0\).
% \end{proof}

The proofs of both corollaries are straightforward from \Cref{thm:random_direction_detection}, by setting \(\psi(\u)=\gamma_{a}(\u)\) and \(\psi(\u)=\tau_q(\u)-\tau_q^0\), respectively. It remains to explain when the local Lipschitz assumption is reasonable for the quantile-based functionals. Rather than proving a general primitive condition here, which would require additional smoothness of the projected density as the direction varies, we state a simple sufficient reduction in terms of the projected quantile curves. 

For the finite-sample theory, we collect all quantile levels used by the skewness and tail diagnostics into a single set. Let \(\mathcal A_S\subset(0,1/2)\) be the finite set of skewness levels and let \(q\in(0,1/2)\) be the (fixed) tail level. Define \(\mathcal A^\star = \{0.50,0.25,0.75,q,1-q\} \cup \{a,1-a:a\in\mathcal A_S\}.\) Thus \(\mathcal A_S\) indexes the skewness functionals, while \(\mathcal A^\star\) is the finite set of quantile levels over which uniform empirical control is required.

\begin{lemma}[Lipschitz quantile curves imply Lipschitz diagnostics]
	\label{lem:lipschitz_bowley_tail}
	Let \(A\subset \mathbb{S}^{p-1}\) be a neighborhood of \(\u_0\). Suppose that, for each \(b\in\mathcal A^\star\), the map \(\u\mapsto Q_b(\u)\) is \(L_b\)-Lipschitz on \(A\) with respect to \(d_{\mathbb S}\). Suppose also that \(Q_{1-a}(\u)-Q_a(\u) \geqslant d_0>0 \) for every \(\u\in A\) and every \(a\in\mathcal A_S\), and that \(Q_{0.75}(\u)-Q_{0.25}(\u)\geqslant d_0>0\) for every \(\u\in A\). Finally, assume \(\max_{b\in\mathcal A^\star}\sup_{\u\in A} \abs{Q_b(\u)} \leqslant M_0. \) Then, for every \(a\in\mathcal A_S\), the map \(\u\mapsto\gamma_{a}(\u)\) is Lipschitz on \(A\), uniformly over \(a\in\mathcal A_S\). The map \(\u\mapsto\tau_q(\u)\) is also Lipschitz on \(A\).
\end{lemma}

\begin{proof}
	Fix \(a\in\mathcal A_S\). Write \(N_a(\u)=Q_{1-a}(\u)+Q_a(\u)-2Q_{0.5}(\u)\), and \(D_a(\u)=Q_{1-a}(\u)-Q_a(\u)\), so that \(\gamma_{a}(\u)=N_a(\u)/D_a(\u)\). By the assumed Lipschitzness of the quantile curves,
	\[
	\abs{N_a(\u)-N_a(\v)} \leqslant \left(L_{1-a}+L_a+2L_{0.50}\right)\,d_{\mathbb S}(\u,\v),
	\]
	and \(\abs{D_a(\u)-D_a(\v)} \leqslant (L_{1-a}+L_a)\,d_{\mathbb S}(\u,\v)\). Also
	\(D_a(\u),D_a(\v)\geqslant d_0\), and \(|N_a(\v)|\leqslant 4M_0\). Hence
	\[
	\abs{\gamma_{a}(\u)-\gamma_{a}(\v)} \leqslant \frac{|N_a(\u)-N_a(\v)|}{d_0} + \frac{|N_a(\v)|\,|D_a(\u)-D_a(\v)|}{d_0^2}.
	\]
	The right-hand side is bounded by a constant multiple of \(d_{\mathbb S}(\u,\v)\). Since \(\mathcal A_S\) is finite and the constants \(L_b\), \(M_0\), and \(d_0\) are uniform over the relevant quantile levels, this gives Lipschitzness of \(\gamma_{a}\) on \(A\), uniformly over
	\(a\in\mathcal A_S\).

	The proof for \(\tau_q\) is identical. Put \(N_\tau(\u)=Q_{1-q}(\u)-Q_q(\u)\) and \(D(\u)=Q_{0.75}(\u)-Q_{0.25}(\u)\), so that \(\tau_q(\u)=N_\tau(\u)/D(\u)\). The assumptions give \(|N_\tau(\u)-N_\tau(v)|\leqslant (L_{1-q}+L_q)d_{\mathbb S}(\u,v)\), \(|D(\u)-D(v)|\leqslant (L_{0.75}+L_{0.25})d_{\mathbb S}(\u,v)\), \(D(\u),D(v)\geqslant d_0\), and \(|N_\tau(v)|\leqslant 2M_0\). Therefore the same quotient bound gives
	\[
	\abs{\tau_q(\u)-\tau_q(v)} \leqslant \frac{|N_\tau(\u)-N_\tau(v)|}{d_0} + \frac{|N_\tau(v)|\,|D(\u)-D(v)|}{d_0^2},
	\]
	which is again bounded by a constant multiple of \(d_{\mathbb S}(\u,v)\). This proves the claim.
\end{proof}

Theorem~\ref{thm:random_direction_detection} gives a useful but incomplete message. Increasing \(m\) improves the chance that at least one sampled direction lands close to a direction carrying signal, but the improvement depends on the spherical-cap probability, which decays quickly with dimension when the cap radius is small. This is the usual geometric cost of random search on a high-dimensional sphere. If asymmetry or tail heaviness is concentrated in a narrow set of directions, random projections may require many trials. This motivates the finite-sample analysis below and, later, the use of structured directions for sparse or otherwise localized alternatives.

\subsection{Uniform finite-sample control over sampled projections}
\label{subsec:finite_sample_control}

The preceding subsection was formulated in terms of population directional functionals. The empirical diagnostics are computed from projected sample quantiles, so we need to control the difference between the empirical and population versions uniformly over the searched directions. The argument is standard: it combines the Dvoretzky--Kiefer--Wolfowitz inequality \citep{dvoretzky1956asymptotic}, Massart's sharp constant \citep{massart1990tight}, and the usual inversion argument for empirical quantiles \citep[see, for example,][Chapter~21]{vandervaart1998asymptotic}. We include the details to make explicit the logarithmic price paid for searching over many directions.

%Let \(\mathcal U_m=\{\u_1,\ldots,\u_m\}\subset \mathbb{S}^{p-1}\) be a finite collection of directions. The directions may be fixed in advance or generated randomly independently of the data. Conditional on \(\mathcal U_m\), the results below are deterministic statements about this finite set. For \(\u\in\mathcal U_m\), let \(Y_i(\u)=\u^\top(\X_i-\boldtheta)\), \(i=1,\ldots,n\), and let \(F_{\u}\) be the distribution function of \(Y_i(\u)\). Write \(Q_a(\u)=\inf\{y:F_{\u}(y)\geqslant a\}\) for the population \(a\)-quantile and \(\widehat Q_a(\u)\) for the empirical \(a\)-quantile of \(Y_1(\u),\ldots,Y_n(\u)\). \sd{are the notations here repetitive and can be shortened?}

The finite-sample analysis treats the location parameter $\boldtheta$ as fixed. For the quantile-based calculations, replacing $\boldtheta$ by a common location estimate shifts all projected observations in direction $\u$ by the same scalar and therefore does not change the corresponding quantile ratios, apart from numerical conventions in the presence of ties. Moment-based skewness and kurtosis require separate treatment of centering error and are not covered by the quantile theory below.

\begin{assumption}[Uniform local quantile regularity]
	\label{ass:uniform_quantile_regular}
	Fix a finite set of skewness levels \(\mathcal A_S\subset(0,1/2)\) and a tail level \(q\in(0,1/2)\). There exist constants \(c_0>0\), \(C_0<\infty\), \(\eta_0>0\), \(d_0>0\) such that, for every \(\u\in\mathcal U_m\) and every \(b\in\mathcal A^\star\), \(F_{\u}(Q_b(\u))=b\), and \(F_{\u}\) has a density \(f_{\u}\) satisfying \(f_{\u}(y)\geqslant c_0\) for all \(y\in[Q_b(\u)-\eta_0,Q_b(\u)+\eta_0]\). In addition, \(\max_{b\in\mathcal A^\star}\max_{\u\in\mathcal U_m}|Q_b(\u)|\leqslant C_0\). The relevant denominators are uniformly bounded away from zero: \(Q_{0.75}(\u)-Q_{0.25}(\u)\geqslant d_0\) for every \(\u\in\mathcal U_m\), and \(Q_{1-a}(\u)-Q_a(\u)\geqslant d_0\) for every \(\u\in\mathcal U_m\) and every \(a\in\mathcal A_S\). All constants \(c_0,C_0,\eta_0\), and \(d_0\) are uniform over \(\u\in\mathcal U_m\) and over the quantile levels in \(\mathcal A^\star\).
\end{assumption}

The lower density condition is the usual quantile-stability condition: it prevents small vertical errors in the empirical distribution function from becoming large horizontal errors in the quantiles. The lower bound on the interquartile range keeps the denominators of the skewness and tail-ratio statistics away from zero.

\begin{lemma}[Uniform empirical quantile control]
	\label{lem:uniform_quantile_control}
	Suppose Assumption~\ref{ass:uniform_quantile_regular} holds. For every \(\eta\in(0,1)\), there is a constant \(C_Q<\infty\), depending only on \(c_0\), such that, with probability at least \(1-\eta\),
	\[
	\max_{u\in\mathcal U_m}\max_{b\in\mathcal A^\star} \abs{\widehat Q_b(\u)-Q_b(\u)} \leqslant C_Q\sqrt{\frac{\log\{2m|\mathcal A^\star|/\eta\}}{n}},
	\]
	provided the right hand side is smaller than \(\eta_0\).
\end{lemma}

\begin{proof}
	Let \(\Delta_n=\{\log(2m|\mathcal A^\star|/\eta)/(2n)\}^{1/2}\). By the Dvoretzky--Kiefer--Wolfowitz inequality with Massart's constant, for each fixed $\u$, \(\Prob\{\sup_y|\widehat F_{\u}(y)-F_{\u}(y)|>\Delta_n\}\leqslant 2\exp(-2n\Delta_n^2)\). A union bound over \(\u\in\mathcal U_m\) gives, with probability at least \(1-\eta\), the event \(\max_{\u\in\mathcal U_m}\sup_y|\widehat F_{\u}(y)-F_{\u}(y)|\leqslant \Delta_n\). We keep the harmless factor \(|\mathcal A^\star|\) in the logarithm so that the later quantile bounds can be written in a uniform notation over the finite set of quantile levels.
	
	Work on this event and fix $\u$ and \(a\). Put \(r_n=2\Delta_n/c_0\). For \(n\) large enough, \(r_n \leqslant \eta_0\). Since \(F_{\u}\) is continuous at \(Q_a(\u)\) and has density at least \(c_0\) in the local neighborhood, \(F_{\u}(Q_a(\u)+r_n)\geqslant a+2\Delta_n\) and \(F_{\u}(Q_a(\u)-r_n)\leqslant a-2\Delta_n\). Therefore \(\widehat F_{\u}(Q_a(\u)+r_n)\geqslant a+\Delta_n\geqslant a\), while \(\widehat F_{\u}(Q_a(\u)-r_n)\leqslant a-\Delta_n<a\). Hence the empirical \(a\)-quantile lies in \([Q_a(\u)-r_n,Q_a(\u)+r_n]\). Thus \(|\widehat Q_a(\u)-Q_a(\u)|\leqslant 2\Delta_n/c_0\). Since the event is uniform over $\u$, the result follows.
\end{proof}

We now translate this quantile bound into control of the directional skewness and tail-ratio statistic. Recall the estimates \(\widehat\gamma_a(\u)\) and \(\widehat\tau_q(\u)\) defined in \eqref{eq:quantile_based_skewness_statistic} and \eqref{eq:quantile_based_tail_statistic}. 

\begin{theorem}[Uniform control of directional quantile skewness]
	\label{thm:uniform_quantile_skewness_concentration}
	Suppose Assumption~\ref{ass:uniform_quantile_regular} holds. Then there is a constant \(C_\gamma<\infty\), depending only on \(c_0\), \(C_0\), and \(d_0\), such that, for every \(\eta\in(0,1)\), with probability at least \(1-\eta\),
	\[
	\max_{\u\in\mathcal U_m}\max_{a\in\mathcal A_S} \abs{\widehat\gamma_{a}(\u)-\gamma_{a}(\u)} \leqslant C_\gamma
	\sqrt{\frac{\log\{2m \abs{\mathcal A^\star}/\eta\}}{n}}.
	\]
\end{theorem}

\begin{proof}
	By Lemma~\ref{lem:uniform_quantile_control}, with probability at least \(1-\eta\), uniformly over \(\u\in\mathcal U_m\), all empirical quantiles indexed by \(\mathcal A^\star\) are within \(\delta_n=C_Q\{\log(2m|\mathcal A^\star|/\eta)/n\}^{1/2}\) of their population counterparts. We work on this event.

	Fix \(\u\in\mathcal U_m\) and \(a\in\mathcal A_S\). Write \(N_a(\u)=Q_{1-a}(\u)+Q_a(\u)-2Q_{0.50}(\u)\) and \(D_a(\u)=Q_{1-a}(\u)-Q_a(\u)\), so that \(\gamma_{a}(\u)=N_a(\u)/D_a(\u)\). Define \(\widehat N_a(\u)\) and \(\widehat D_a(\u)\) analogously. The quantile-control event gives \(|\widehat N_a(\u)-N_a(\u)|\leqslant 4\delta_n\) and \(|\widehat D_a(\u)-D_a(\u)|\leqslant 2\delta_n\). By Assumption~\ref{ass:uniform_quantile_regular}, \(D_a(\u)\geqslant d_0\), and, for \(n\) large enough, \(\widehat D_a(\u)\geqslant d_0/2\). Also \(|N_a(\u)|\leqslant 4C_0\).

	Using
	\[
	\left|
	\frac{\widehat N_a(\u)}{\widehat D_a(\u)}
	-
	\frac{N_a(\u)}{D_a(\u)}
	\right|
	\le
	\frac{|\widehat N_a(\u)-N_a(\u)|}{\widehat D_a(\u)}
	+
	|N_a(\u)|
	\frac{|\widehat D_a(\u)-D_a(\u)|}{\widehat D_a(\u)D_a(\u)},
	\]
	we obtain
	\[
    \abs{\widehat\gamma_{a}(\u)-\gamma_{a}(\u)} \leqslant \left(\frac{8}{d_0}+\frac{16C_0}{d_0^2}\right)\delta_n .
	\]
	The bound is uniform over \(\u\in\mathcal U_m\) and \(a\in\mathcal A_S\). Thus one may take \(C_\gamma=C_Q(8/d_0+16C_0/d_0^2)\), up to an absolute numerical factor.
\end{proof}

\begin{theorem}[Uniform control of directional tail ratios]
	\label{thm:uniform_tail_concentration}
	Suppose Assumption~\ref{ass:uniform_quantile_regular} holds. Then there is a constant \(C_\tau<\infty\), depending only on \(c_0\), \(C_0\), and \(d_0\), such that, for every \(\eta\in(0,1)\), with probability at least \(1-\eta\),
	\[
	\max_{u\in\mathcal U_m} \abs{\widehat\tau_q(\u)-\tau_q(\u)} \leqslant C_\tau \sqrt{\frac{\log\{2m|\mathcal A^\star|/\eta\}}{n}}.
	\]
	Consequently, for any fixed reference value \(\tau_q^0\),
	\[
	\abs{\max_{u\in\mathcal U_m} \abs{\widehat\tau_q(\u)-\tau_q^0} - \max_{u\in\mathcal U_m} \abs{\tau_q(\u)-\tau_q^0}} \leqslant C_\tau \sqrt{\frac{\log\{2m|\mathcal A^\star|/\eta\}}{n}}
	\]
	with probability at least \(1-\eta\).
\end{theorem}

\begin{proof}
	Denote the right hand side of the last equation as $\delta_n$. Apply Lemma~\ref{lem:uniform_quantile_control} and work on the event that all relevant empirical quantiles are within \(\delta_n\) of their population values. For fixed $\u$, write \(N_T(\u)=Q_{1-q}(\u)-Q_q(\u)\), \(D(\u)=Q_{0.75}(\u)-Q_{0.25}(\u)\), and define \(\widehat N_T(\u)\), \(\widehat D(\u)\) analogously. Then \(|\widehat N_T(\u)-N_T(\u)|\leqslant 2\delta_n\), \(|\widehat D(\u)-D(\u)| \leqslant 2\delta_n\), \(D(\u)\geqslant d_0\), \(\widehat D(\u)\geqslant d_0/2\) for \(n\) large enough, and \(|N_T(\u)|\leqslant 2C_0\). Therefore
	\[
	\abs{\widehat\tau_q(\u)-\tau_q(\u)} \leqslant \frac{2\delta_n}{d_0/2} + 2C_0\frac{2\delta_n}{(d_0/2)d_0} = \left(\frac{4}{d_0}+\frac{8C_0}{d_0^2}\right)\delta_n .
	\]
	Take \(C_\tau=C_Q\{4/d_0+8C_0/d_0^2\}\), up to an absolute numerical factor.	This proves the first claim uniformly over $\u$.
	
	For the second claim, use the elementary inequality
	\(\abs{\max_u|a_u|-\max_u|b_u|} \leqslant \max_u|a_u-b_u|\), with \(a_u=\widehat\tau_q(\u)-\tau_q^0\) and \(b_u=\tau_q(\u)-\tau_q^0\). 
\end{proof}

\begin{corollary}[Uniform control of the maximum skewness diagnostic]
	\label{cor:max_skewness_concentration}
	Under the assumptions of Theorem~\ref{thm:uniform_quantile_skewness_concentration}, with probability at least \(1-\eta\),
	\[
	\abs{\max_{\u\in\mathcal U_m}\max_{a\in\mathcal A_S} \abs{\widehat\gamma_{a}(\u)} - \max_{\u\in\mathcal U_m}\max_{a\in\mathcal A_S}\abs{\gamma_{a}(\u)}} \leqslant C_\gamma \sqrt{\frac{\log\{2m|\mathcal A^\star|/\eta\}}{n}}.
	\]
\end{corollary}

\begin{proof}
	Use the elementary inequality
	\(
	\left|\max_i |x_i|-\max_i |y_i|\right|
	\leqslant
	\max_i |x_i-y_i|
	\)
	over the finite index set \(i=(\u,a)\in\mathcal U_m\times\mathcal A_S\), with \(x_i=\widehat\gamma_{a}(\u)\) and \(y_i=\gamma_{a}(\u)\). The claim then follows directly from Theorem~\ref{thm:uniform_quantile_skewness_concentration}.
\end{proof}

Theorems~\ref{thm:uniform_quantile_skewness_concentration} and \ref{thm:uniform_tail_concentration} show that the empirical projection diagnostics approximate their population counterparts uniformly over \(\mathcal U_m\) at the rate \(\sqrt{\log(m)/n}\), up to constants and the fixed number of quantile levels involved. Thus, searching over more directions improves the chance of finding localized departures, but increases the stochastic fluctuation of the maximum only logarithmically. The quantile formulation is also important: no third or fourth moment assumption is required. The bounds remain meaningful for heavy-tailed distributions for which moment-based skewness or kurtosis may be unstable or undefined.

\subsection{Consistency of the four-regime diagnostic}
\label{subsec:four_regime_consistency}

We now connect the uniform concentration bounds to the four-regime interpretation of the diagnostic. The result is stated relative to the searched direction set \(\mathcal U_m\). Thus the population quantities below measure only the skewness and tail-ratio evidence visible to the chosen projections. If \(\mathcal U_m\) is fixed, the conclusion is a consistency statement for that finite projection family; if \(m=m_n\) grows, it applies to the corresponding growing family, provided \(\log m_n=o(n)\).

% Recall the population skewness summary \(S_{\mathcal U,\mathcal A_S}=\max_{\u\in\mathcal U_m}\max_{a\in\mathcal A_S}|\gamma_{a}(\u)|\). For the tail component, we use the two-sided tail-departure summary \(T_{\mathcal U}=\max_{\u\in\mathcal U_m}|\tau_q(\u)-\tau_q^0|\), where \(\tau_q^0\) is the chosen reference tail value, under Gaussian calibration or otherwise. The empirical analogues are \(\widehat S_{\mathcal U_m,\mathcal A_S}=\max_{\u\in\mathcal U_m}\max_{a\in\mathcal A_S}|\widehat\gamma_{a}(\u)|\) and \(\widehat T_{\mathcal U_m}=\max_{\u\in\mathcal U_m}|\widehat\tau_q(\u)-\tau_q^0|\). \sd{again, maybe we should simply refer to the original test statistics defined in Sec 2? In fact, we need to reduce overall page count -- so maybe the below discussion can be reduced? I think some parts of it are already mentioned in the previous sections.}

The four regimes are defined through whether \(S_{\mathcal U_m,\mathcal A_S}\) and \(T_{\mathcal U_m}\), given by expressions \eqref{eq:skewness_measure} and \eqref{eq:tail_measure} respectively, are null or separated from zero. For positive constants \(\Delta_S\) and \(\Delta_T\), we consider the classes
\[
\mathcal C_0:\ S_{\mathcal U_m,\mathcal A_S}=0,\ T_{\mathcal U_m}=0,\qquad
\mathcal C_T:\ S_{\mathcal U_m,\mathcal A_S}=0,\ T_{\mathcal U_m}\ge\Delta_T,
\]
and
\[
\mathcal C_S:\ S_{\mathcal U_m,\mathcal A_S}\ge\Delta_S,\ T_{\mathcal U_m}=0,\qquad
\mathcal C_{ST}:\ S_{\mathcal U_m,\mathcal A_S}\ge\Delta_S,\ T_{\mathcal U_m}\ge\Delta_T.
\]
These correspond, respectively, to symmetric reference-tail, symmetric tail-departed, skewed reference-tail, and skewed tail-departed behavior relative to the chosen tail reference. When the reference is Gaussian and the observed tail-ratio exceeds \(\tau_q^0\), the tail-departed regimes may be interpreted as heavier-tailed than Gaussian. However, because \(T_{\mathcal U_m}\) is two-sided, it also includes projections whose tail-ratio is smaller than the reference value. Exact zeros are useful for exposition; an approximate version is given at the end of the subsection.

\begin{remark}
	All four regimes in this subsection are defined relative to the searched direction set \(\mathcal U_m\) and the finite skewness-level set \(\mathcal A_S\). Thus \(S_{\mathcal U_m,\mathcal A_S}=0\) means that no quantile skewness is visible among the inspected projections and the chosen levels \(a\in\mathcal A_S\), not necessarily that \(\gamma_{a}(\u)=0\) for every \(\u\in\mathbb S^{p-1}\) and every \(a\in(0,1/2)\). Similarly, \(T_{\mathcal U_m}=0\) means that every inspected projection has the same tail-ratio as the chosen reference, not necessarily that the full distribution has reference-law tails. A global version would require additional control over the full sphere and, if desired, over a continuum of skewness levels.
\end{remark}

Let \(a_{n,m}(\eta)=C\{\log(2m|\mathcal A^\star|/\eta)/n\}^{1/2}\), where \(C\) is large enough to dominate the constants in Theorems~\ref{thm:uniform_quantile_skewness_concentration} and \ref{thm:uniform_tail_concentration}. Theorem~\ref{thm:uniform_quantile_skewness_concentration} controls \(|\widehat S_{\mathcal U_m,\mathcal A_S}-S_{\mathcal U_m,\mathcal A_S}|\), while Theorem~\ref{thm:uniform_tail_concentration} controls \(|\widehat T_{\mathcal U_m}-T_{\mathcal U_m}|\), both at the rate \(a_{n,m}(\eta)\), where \(\widehat S_{\mathcal U_m,\mathcal A_S}\) and \(\widehat T_{\mathcal U_m}\) are given by \eqref{eq:empirical_skew_tail_measure}. We now present a deterministic lemma showing that a threshold classifier is stable when the empirical summaries are closer to their population values than the relevant separation margins.

\begin{lemma}[Deterministic classification stability]
	\label{lem:deterministic_classification_stability}
	Suppose \(|\widehat S_{\mathcal U_m,\mathcal A_S}-S_{\mathcal U_m,\mathcal A_S}|\le\varepsilon_S\) and \(|\widehat T_{\mathcal U_m}-T_{\mathcal U_m}|\le\varepsilon_T\). Consider the rule
	\[
	\widehat C =
	\begin{cases}
		\mathcal C_0, & \widehat S_{\mathcal U_m,\mathcal A_S}<\lambda_S,\ \widehat T_{\mathcal U_m}<\lambda_T,\\
		\mathcal C_T, & \widehat S_{\mathcal U_m,\mathcal A_S}<\lambda_S,\ \widehat T_{\mathcal U_m}\ge\lambda_T,\\
		\mathcal C_S, & \widehat S_{\mathcal U_m,\mathcal A_S}\ge\lambda_S,\ \widehat T_{\mathcal U_m}<\lambda_T,\\
		\mathcal C_{ST}, & \widehat S_{\mathcal U_m,\mathcal A_S}\ge\lambda_S,\ \widehat T_{\mathcal U_m}\ge\lambda_T.
	\end{cases}
	\]
	If \(0<\lambda_S<\Delta_S\), \(0<\lambda_T<\Delta_T\), \(\varepsilon_S<\min\{\lambda_S,\Delta_S-\lambda_S\}\), and \(\varepsilon_T<\min\{\lambda_T,\Delta_T-\lambda_T\}\), then the rule assigns the correct population class among \(\mathcal C_0,\mathcal C_T,\mathcal C_S,\mathcal C_{ST}\).
\end{lemma}

\begin{proof}
	If the population class is \(\mathcal C_0\), then \(S_{\mathcal U_m,\mathcal A_S}=0\) and \(T_{\mathcal U_m}=0\). Hence \(\widehat S_{\mathcal U_m,\mathcal A_S}\le\varepsilon_S<\lambda_S\) and \(\widehat T_{\mathcal U_m}\le\varepsilon_T<\lambda_T\), so the rule assigns \(\mathcal C_0\).

	If the population class is \(\mathcal C_T\), then \(S_{\mathcal U_m,\mathcal A_S}=0\) and \(T_{\mathcal U_m}\ge\Delta_T\). Thus \(\widehat S_{\mathcal U_m,\mathcal A_S}\le\varepsilon_S<\lambda_S\), while \(\widehat T_{\mathcal U_m}\geqslant T_{\mathcal U_m}-\varepsilon_T\ge\Delta_T-\varepsilon_T>\lambda_T\). Hence the rule assigns \(\mathcal C_T\).

	If the population class is \(\mathcal C_S\), then \(S_{\mathcal U_m,\mathcal A_S}\ge\Delta_S\) and \(T_{\mathcal U_m}=0\). Thus \(\widehat S_{\mathcal U_m,\mathcal A_S}\geqslant S_{\mathcal U_m,\mathcal A_S}-\varepsilon_S\ge\Delta_S-\varepsilon_S>\lambda_S\), while \(\widehat T_{\mathcal U_m}\le\varepsilon_T<\lambda_T\). Hence the rule assigns \(\mathcal C_S\).

	Finally, if the population class is \(\mathcal C_{ST}\), then \(S_{\mathcal U_m,\mathcal A_S}\ge\Delta_S\) and \(T_{\mathcal U_m}\ge\Delta_T\). Therefore \(\widehat S_{\mathcal U_m,\mathcal A_S}\ge\Delta_S-\varepsilon_S>\lambda_S\) and \(\widehat T_{\mathcal U_m}\ge\Delta_T-\varepsilon_T>\lambda_T\), so the rule assigns \(\mathcal C_{ST}\).
\end{proof}

The theoretical classifier below uses deterministic thresholds on \(\widehat S_{\mathcal U_m,\mathcal A_S}\) and \(\widehat T_{\mathcal U_m}\). This is the formal object for which we prove finite-sample correctness. In the implementation, these thresholds are replaced by simulation-calibrated reference quantiles, equivalently by calibrated \(p\)-values as given in \eqref{eq:calibrated_p_val}.

\begin{theorem}[Finite-sample correctness under separated regimes]
	\label{thm:finite_sample_four_regime}
	Suppose Assumption~\ref{ass:uniform_quantile_regular} holds. Suppose the population distribution belongs to one of \(\mathcal C_0,\mathcal C_T,\mathcal C_S,\mathcal C_{ST}\) relative to the direction set \(\mathcal U_m\), the skewness-level set \(\mathcal A_S\), and the chosen tail reference \(\tau_q^0\). Let \(\widehat C\) be the threshold rule in Lemma~\ref{lem:deterministic_classification_stability}. If \(0<\lambda_S<\Delta_S\) and \(0<\lambda_T<\Delta_T\), then, for every \(\eta\in(0,1)\), the classification is correct with probability at least \(1-\eta\), provided \(a_{n,m}(\eta)<\min\{\lambda_S,\Delta_S-\lambda_S,\lambda_T,\Delta_T-\lambda_T\}\). In particular, with \(\lambda_S=\Delta_S/2\) and \(\lambda_T=\Delta_T/2\), it is enough that \(a_{n,m}(\eta)<\frac12\min\{\Delta_S,\Delta_T\}\).
\end{theorem}

\begin{proof}
	By Corollary~\ref{cor:max_skewness_concentration} and Theorem~\ref{thm:uniform_tail_concentration}, after increasing the constant in \(a_{n,m}(\eta)\) if necessary, the event \(|\widehat S_{\mathcal U_m,\mathcal A_S}-S_{\mathcal U_m,\mathcal A_S}|\leqslant a_{n,m}(\eta)\) and \(|\widehat T_{\mathcal U_m}-T_{\mathcal U_m}|\leqslant a_{n,m}(\eta)\) has probability at least \(1-\eta\). On this event, the assumed sample-size condition implies the margin condition in Lemma~\ref{lem:deterministic_classification_stability}, with \(\varepsilon_S=\varepsilon_T=a_{n,m}(\eta)\). The conclusion follows.
\end{proof}

\begin{corollary}[Asymptotic consistency over a growing projection family]
	\label{cor:asymptotic_four_regime_consistency}
	Let \(m=m_n\) and suppose Assumption~\ref{ass:uniform_quantile_regular} holds uniformly in \(n\). Suppose also that \(\mathcal A_S\) is fixed. If \(\log m_n=o(n)\), and the population distribution belongs to one of the four separated classes above with fixed separation constants \(\Delta_S>0\) and \(\Delta_T>0\), then, for \(\lambda_S=\Delta_S/2\) and \(\lambda_T=\Delta_T/2\), \(\Prob(\widehat C=C)\to1\), where \(C\) denotes the population class relative to \(\mathcal U_{m_n}\), \(\mathcal A_S\), and the chosen tail reference \(\tau_q^0\).
\end{corollary}

\begin{proof}
	For every fixed \(\eta\in(0,1)\), since \(\mathcal A_S\) is fixed, \(|\mathcal A^\star|\) is fixed, and the condition \(\log m_n=o(n)\) implies \(a_{n,m_n}(\eta)=C\{\log(2m_n|\mathcal A^\star|/\eta)/n\}^{1/2}\to0\). Hence, for all large \(n\), \(a_{n,m_n}(\eta)<\frac12\min\{\Delta_S,\Delta_T\}\). The finite-sample theorem then gives \(\Prob(\widehat C=C)\ge1-\eta\). Since \(\eta\) is arbitrary, the result follows.
\end{proof}

The uniformity requirement in Assumption~\ref{ass:uniform_quantile_regular} is substantive when \(m=m_n\) grows. It requires the local density lower bound, the relevant quantile bounds, and the denominator lower bounds \(Q_{1-a}(\u)-Q_a(\u)\geqslant d_0\), \(a\in\mathcal A_S\), and \(Q_{0.75}(\u)-Q_{0.25}(\u)\geqslant d_0\), to hold uniformly over the growing direction set. For fixed \(m\), this is mild. For increasing random direction sets, it should be understood as a regularity condition on the underlying distribution and on the range of projections being searched.

The preceding result uses deterministic thresholds because this is the cleanest way to express the separation argument. The implemented procedure may instead use simulation-calibrated \(p\)-values. The same principle is at work: under a reference law, the calibration distribution accounts for the finite-sample fluctuation of the maxima over \(\mathcal U_m\) and \(\mathcal A_S\); under separated alternatives, the population signal is bounded away from the reference value, while the empirical error is of order \(\{\log(m|\mathcal A^\star|)/n\}^{1/2}\). Once the signal is larger than the calibration fluctuation, the corresponding calibrated \(p\)-value tends to be small. Under the reference law, calibration controls the nominal error up to Monte Carlo error and model misspecification.

\begin{remark}[Approximate regimes]
	\label{rem:approximate_regimes}
	The exact equalities \(S_{\mathcal U_m,\mathcal A_S}=0\) and \(T_{\mathcal U_m}=0\) are idealizations. One may instead define approximate null and alternative regions by \(S_{\mathcal U_m,\mathcal A_S}\le\delta_S\) versus \(S_{\mathcal U_m,\mathcal A_S}\ge\Delta_S\), and \(T_{\mathcal U_m}\le\delta_T\) versus \(T_{\mathcal U_m}\ge\Delta_T\), where \(0\le\delta_S<\Delta_S\) and \(0\le\delta_T<\Delta_T\). The same proof applies whenever the thresholds satisfy \(\delta_S<\lambda_S<\Delta_S\) and \(\delta_T<\lambda_T<\Delta_T\), with empirical errors smaller than the minimum margin
	\[
	\min\{\lambda_S-\delta_S,\Delta_S-\lambda_S,\lambda_T-\delta_T,\Delta_T-\lambda_T\}.
	\]
	This approximate formulation is more realistic in applications, where a distribution may be nearly symmetric, may show weak quantile skewness only at some levels \(a\in\mathcal A_S\), or may have tail-ratios only mildly different from the reference law.
\end{remark}

\subsection{Sparse directional alternatives}
\label{subsec:sparse_directional_alternatives}

The preceding results show that random projections can detect a localized directional departure if at least one sampled direction falls close to a signal-bearing direction. In high dimensions, however, a purely random search can be inefficient when the departure is concentrated in a small number of coordinates. We illustrate this point through a simple rank-one alternative. Consider
\begin{equation}
    \label{eq:rank_one_alt}
    \X = \bm{Z} + \delta V \bm{v},
\end{equation}
where \(\bm{Z}\sim N_p(0,\identity_p)\), \(\v\in \mathbb{S}^{p-1}\), \(V\) is a scalar random variable independent of \(\bm Z\), and \(\delta>0\) controls the strength of the non-Gaussian perturbation. The vector \(\v\) is the signal direction. If \(V\) is skewed, the perturbation creates directional asymmetry. If \(V\) is symmetric but heavy-tailed, it creates tail inflation without directional skewness. For any \(\u\in \mathbb{S}^{p-1}\), \(\u^\top \X = \u^\top \bm Z+\delta(\u^\top \v)V\). Thus the non-Gaussian component enters the projection only through the alignment \(\u^\top \v\).

The calculations below are stated in terms of cumulants. They should be read as a transparent proxy for directional skewness and tail inflation. The main diagnostic in the paper uses quantile-based summaries, but the cumulant model makes the high-dimensional geometry particularly clear.

\begin{proposition}[Rank-one skewness signal]
	\label{prop:rank_one_third_cumulant}
	Suppose \eqref{eq:rank_one_alt} holds, \(\E(V)=0\), \(\E|V|^3<\infty\), and \(\kappa_3(V)\ne0\), where \(\kappa_3(V)=\E(V^3)\). Then, for every \(\u\in \mathbb{S}^{p-1}\), $\kappa_3(\u^\top \X)=\delta^3 (\u^\top \v)^3\kappa_3(V)$. Consequently, \(\sup_{\|\u\| =1} \abs{\kappa_3(\u^\top \X)} = \delta^3 \abs{\kappa_3(V)}\), and the supremum is attained at \(\u=\pm \v\).
\end{proposition}

\begin{proof}
	For fixed $\u$, \(\u^\top \X=\u^\top \bm Z+\delta(\u^\top \v)V\). The two terms are independent, \(\u^\top \bm Z \sim N(0,1)\), and the third cumulant of a centered Gaussian variable is zero. Since cumulants are additive over independent sums and homogeneous of order three, \(\kappa_3(\u^\top \X)=\delta^3(\u^\top \v)^3\kappa_3(V)\). The supremum follows from \(|\u^\top \v| \leqslant 1\).
\end{proof}

The above indicates that the optimal direction is \(\v\), but a random search does not know \(\v\). Its efficacy is governed by \(M_m(\v)=\max_{1\leqslant j\leqslant m}|\u_j^\top \v|\), where \(\u_1,\ldots,\u_m\) are independent uniform directions on \(\mathbb{S}^{p-1}\). By rotational invariance, the distribution of \(|\u_j^\top \v|\) is same as that of the absolute first coordinate of a uniformly distributed point on the sphere. 

The following standard fact is a consequence of concentration of measure on the sphere; see, for example, \citet{ledoux2001concentration} and \citet[Chapter~3]{vershynin2018highdimensional}.

\begin{lemma}[Random alignment with a fixed direction]
	\label{lem:random_alignment_order}
	Let \(\u_1,\ldots,\u_m\) be independent uniform directions on \(\mathbb{S}^{p-1}\), and let \(\v\in \mathbb{S}^{p-1}\) be fixed. There exist universal constants \(c,C>0\) such that, for \(0<\eta<1\), with probability at least \(1-\eta\),
	\(M_m(\v)\leqslant C\sqrt{\log(2m/\eta)/p}.\)
	Conversely, there exist universal constants \(c,C>0\) such that, whenever
	\(1<m\leqslant \exp(cp)\), for some universal \(\alpha\in(0,1)\),
	\(
	\Prob\left\{
	M_m(\v)\geqslant c\sqrt{(\log m)/p}
	\right\}
	\geqslant 1-\exp\left\{-C m^{\alpha}\right\}.
	\)
\end{lemma}

\begin{proof}
	The upper bound follows from the standard spherical-coordinate tail bound \(\Prob(|\bm{U}_1|\geqslant t)\leqslant 2\exp(-cpt^2)\) for \(\bm{U}\sim{\rm Unif}(\mathbb{S}^{p-1})\), followed by a union bound over \(m\) directions. The lower bound follows from the matching lower bound for the spherical cap probability and the identity \(\Prob\{M_m(\v)<t\}=\{1-\Prob(|\bm{U}_1|\geqslant t)\}^m\). These are standard estimates for spherical caps and random projections.
\end{proof}

This lower bound is informative when \(\log m\) is non-negligible relative to the scale of the search problem. In particular, for polynomially many random directions, \(m=p^b\), the typical alignment is of order \(\sqrt{\log p/p}\); detecting very sparse alternatives may therefore require augmenting random directions with coordinate or structured directions.

Combining the cumulant identity with the alignment bound gives the random-projection signal size.

\begin{theorem}[Random-projection skewness signal under the rank-one alternative]
	\label{thm:rank_one_random_projection_skew_signal}
	Under the assumptions of Proposition~\ref{prop:rank_one_third_cumulant}, if \(\u_1,\ldots,\u_m\) are independent uniform directions, then
	\(
	\max_{1\leqslant j\leqslant m} \abs{\kappa_3\left(\u_j^\top \X\right)} = \delta^3 \abs{\kappa_3(V)} \,M_m(\v)^3.
	\)
	In particular, for every \(\eta\in(0,1)\),
	\[
	\max_{1\leqslant j\leqslant m} \abs{\kappa_3(\u_j^\top \X)} \leqslant C\delta^3 \abs{\kappa_3(V)} \left\{\frac{\log(2m/\eta)}{p}\right\}^{3/2}
	\]
	with probability at least \(1-\eta\). Conversely, in the regime covered by the lower bound in Lemma~\ref{lem:random_alignment_order}, the same quantity is bounded below, with the probability stated there, by a constant multiple of $\delta^3 \abs{\kappa_3(V)} \left(p^{-1}\log m\right)^{3/2}$.
\end{theorem}

\begin{proof}
	By Proposition~\ref{prop:rank_one_third_cumulant}, \(|\kappa_3(\u_j^\top \X)|=\delta^3|\kappa_3(V)|\,|\u_j^\top \v|^3\). Taking the maximum over \(j\) gives the identity. The order statement follows from Lemma~\ref{lem:random_alignment_order}.
\end{proof}

This order can be small when \(p\) is large unless \(m\) is very large, as shown next.

\begin{proposition}[Coordinate augmentation for sparse signal directions]
	\label{prop:coordinate_sparse_signal}
	Suppose \(\v\in \mathbb{S}^{p-1}\) is \(s\)-sparse, meaning \(\|\v\|_0\leqslant s\). Then \(\max_{1\leqslant k\leqslant p}|\bm{e}_k^\top \v|\geqslant s^{-1/2}\). Consequently, under \eqref{eq:rank_one_alt} and the assumptions of Proposition~\ref{prop:rank_one_third_cumulant},
	\[
	\max_{1\leqslant k\leqslant p} \abs{\kappa_3\left(\bm e_k^\top \X\right)} 	\geqslant \delta^3 \abs{\kappa_3(V)}s^{-3/2}.
	\]
\end{proposition}

\begin{proof}
	Since \(\v\) is \(s\)-sparse and \(\|\v\|_2=1\), at least one nonzero coordinate must have magnitude at least \(s^{-1/2}\). Applying Proposition~\ref{prop:rank_one_third_cumulant} with \(\u= \bm e_k\) gives the stated lower bound.
\end{proof}

A parallel calculation applies to tail inflation when \(V\) has a positive fourth cumulant. This result is given below.

\begin{proposition}[Rank-one tail signal]
	\label{prop:rank_one_fourth_cumulant}
	Suppose \eqref{eq:rank_one_alt} holds, \(\E(V)=0\), \(\E(V^4)<\infty\), and \(\kappa_4(V)=\E(V^4)-3\{\E(V^2)\}^2>0\). Then, for \(\u\in \mathbb{S}^{p-1}\), $\kappa_4(\u^\top \X)=\delta^4(\u^\top \v)^4\kappa_4(V)$.	Consequently, \(\sup_{\|\u\|=1}|\kappa_4(\u^\top \X)|=\delta^4\kappa_4(V)\), attained at \(\u=\pm \v\).
\end{proposition}

\begin{proof}
	The proof is same as the previous case. Fourth cumulants are additive over independent sums and are homogeneous of order four. The fourth cumulant of \(\u^\top \bm Z \sim N(0,1)\) is 0, so only the perturbation term contributes.
\end{proof}

\begin{corollary}[Random and coordinate signal sizes for fourth cumulants]
	\label{cor:rank_one_tail_signal_sizes}
	Under assumptions of Proposition~\ref{prop:rank_one_fourth_cumulant}, $\max_{1\leqslant j\leqslant m}|\kappa_4(\u_j^\top \X)| = \delta^4\kappa_4(V)M_m(\v)^4$. Hence, the random-projection fourth-cumulant signal is of order \(\delta^4\kappa_4(V)(\log m/p)^2\) with high probability. If \(\v\) is \(s\)-sparse, then the coordinate directions satisfy	$\max_{1\leqslant k\leqslant p}|\kappa_4(\bm e_k^\top \X)| \geqslant \delta^4\kappa_4(V)s^{-2}$.
\end{corollary}

\begin{proof}
	The first identity follows from Proposition~\ref{prop:rank_one_fourth_cumulant}. The random-direction order follows from Lemma~\ref{lem:random_alignment_order}. The coordinate lower bound follows from \(\max_k|\bm e_k^\top \v|\geqslant s^{-1/2}\).
\end{proof}

These calculations make the role of the direction set explicit. For skewness, a random search sees a signal of order \(\delta^3|\kappa_3(V)|(\log m/p)^{3/2}\), whereas coordinate directions see at least \(\delta^3|\kappa_3(V)|s^{-3/2}\) when the signal direction is \(s\)-sparse. For fourth-cumulant tail inflation, the analogous orders are \(\delta^4\kappa_4(V)(\log m/p)^2\) for random directions and \(\delta^4\kappa_4(V)s^{-2}\) for coordinate directions. Thus random directions provide broad model-free coverage of the sphere, but can be inefficient for sparse high-dimensional alternatives. Coordinate augmentation is beneficial precisely in such cases. We find it important to note that directions obtained from PCA or ICA may also improve sensitivity when the non-Gaussian direction is aligned with a dominant variance component or an independent component. Their formal treatment, however, requires either sample splitting or additional perturbation arguments, because those directions are data-dependent. We therefore leave PCA and ICA directions as possible future work, while the theory above explains why random and coordinate directions already capture two complementary regimes.

\section{Simulation study}\label{sec:simulation}

We conduct a simulation study to evaluate the finite-sample behavior of the proposed projection diagnostics.  The simulation is designed to assess three aspects of the method: empirical size under a Gaussian benchmark, power against symmetric heavy-tailed alternatives, and power against skewed alternatives with and without heavy tails.  We also examine the effect of dimension, sample size, and the magnitude of directional skewness. All of our empirical analysis are done in RStudio (version 2025.05.0+496), equipped with R version 4.5.1 running under macOS Tahoe 26.5.

\subsection{Setups}

For all simulation settings, observations are generated independently in $\R^p$.  Dependence among coordinates is induced through the covariance matrix $\Sigma$ with the $(i,j)^{th}$ element $\rho^{|i-j|})$ where $\rho$ controls the extent of dependence. We found the results to be mostly robust for different choices of $\rho$ and shall report the results for $\rho=0.5$ below. In line with our earlier discussions, each generated data matrix is centered columnwise before the projection statistics are computed; and the diagnostic directions are taken to be $\mathcal U_m = \mathcal U_{\mathrm{rand}} \cup \mathcal U_{\mathrm{coord}}$, where $\mathcal U_{\mathrm{rand}}$ consists of $m$ independent random directions drawn uniformly from $\mathbb S^{p-1}$, and $\mathcal U_{\mathrm{coord}}=\{\bm{e}_1,\ldots,\bm{e}_p\}$ is the set of coordinate directions.  In the implementation, a random direction is generated as $\bm{Z}/\norm{\bm Z}_2$, where $\bm Z \sim N_p(0,\identity_p)$.  The coordinate directions are included because they improve sensitivity to sparse directional departures.

The skewness diagnostic is based on several quantile levels rather than only the Bowley level. Specifically, following \eqref{eq:quantile_based_skewness_statistic}, we take the statistic
\begin{equation*}
        \widehat S_{\mathcal U_m,\mathcal A_S} = \max_{\u\in\mathcal U_m} \max_{a \in\mathcal A_S}
        \abs{\widehat\gamma_{a}(\u)}.
\end{equation*}
with the grid $\mathcal A_S = \{0.05,0.10,0.15,0.20,0.25\}$. We reiterate that the case $a=0.25$ corresponds to the Bowley-type skewness, while smaller values of $a$ allow the statistic to detect asymmetry away from the central inter-quartile region.

On the other hand, the tail diagnostic is based on the inter-quantile tail ratio $\widehat\tau_{a}(\u)$ defined in \eqref{eq:quantile_based_tail_statistic}. With the Gaussian benchmark value defined as $\tau_q^N$, we use the one-sided statistic
\begin{equation*}
        \widehat T_{\mathcal U_m}^{+} = \max_{\u\in\mathcal U_m} \left\{\widehat\tau_{q}(\u) - \tau_q^N\right\}_{+},
\end{equation*}
where $x_+=\max(x,0)$ and $q$ is fixed at 0.025 which is found to render good results, thereby obviating the need to take a grid of multiple quantile levels. It is also imperative to point out that this statistic helps with the objective of detecting heavy-tailedness rather than any tail-shape deviation. 

For both tests, critical values are obtained by Monte Carlo calibration under the Gaussian reference model.  For each pair $(n,p)$ and for the same direction set $\mathcal U_m$, we generate $B$ independent samples from $N_p(0,\identity_p)$, center each sample, and compute the corresponding values of $\widehat S_{\mathcal U_m,\mathcal A_S}$ and $\widehat T_{\mathcal U_m}^{+}$.  The empirical $(1-\alpha)$ quantiles of the resulting null distributions are used as critical values, with $\alpha=0.05$. The same quantile grid, tail probability, direction set construction, and calibration scheme are used for all competing data-generating mechanisms.

For the data-generating mechanisms, we consider five different settings.  The first is the Gaussian benchmark, where observations are generated independently from $N_p(0,\Sigma)$. This model represents the symmetric benchmark-tailed regime and is used to evaluate empirical size. The second model is the multivariate $t$ distribution generated through the scale-mixture representation $\bm{Z}/\sqrt{W}$, where $\bm{Z} \sim N_p(0,\Sigma)$, $W\sim \chi^2_{\nu}/\nu$ and they are independent of each other.  We take $\nu=5$ in our simulations. This model is centrally symmetric but heavy-tailed, and hence represents the symmetric heavy-tailed regime. The third model is a hidden-factor skew-normal-type model, where the sample observations are generated as $\bm{Z}+\boldalpha\abs{H}$, with $\bm{Z} \sim N_p(0,\Sigma)$ independent of $H \sim N(0,1)$. The vector $\boldalpha \in \R^p$ determines the direction and magnitude of asymmetry. This model is used to represent the skewed benchmark-tailed regime.  For a projection direction $\u$, the projected variable satisfies $\u^\top \bm{Z} + (\u^\top\boldalpha)\abs{H}$, so the strength of projected skewness depends on the alignment between $\u$ and $\boldalpha$. The fourth model is a hidden-factor skew-$t$-type model where the random variable mentioned in the third model is divided by $\sqrt{W}$, for an independent variable $W\sim \chi^2_{\nu}/\nu$, in the same spirit as multivariate $t$ distribution. We again use $\nu=5$ in this.  This model combines directional asymmetry with heavy tails and therefore represents the skewed heavy-tailed regime. Finally, the fifth model is a symmetric contaminated normal distribution, where each observation is generated from $N_p(0,\Sigma)$ with probability $1-\varepsilon$, and from $N_p(0,c\Sigma)$ with probability $\varepsilon$. We use $\varepsilon=0.05$ and $c=10$ in our implementation. Observe that this model is symmetric about zero but has inflated tails due to a small fraction of high-variance observations.  It is included to assess whether the tail diagnostic detects contamination-driven tail inflation without producing excessive skewness rejections.

For the skewed models, we first use $\boldalpha = \delta p^{-0.5} \, \mathbf{1}_p$, where $\mathbf{1}_p$ is a vector of all 1, to introduce skewness evenly in all directions. Then, we introduce sparse directional skewness with $\boldalpha = \delta \bm{e}_1$, where $\bm{e}_1$ is the first coordinate vector. In all cases, $\delta>0$ is the skewness-strength parameter. It is not itself a skewness coefficient, but it is the magnitude of the asymmetric latent-factor perturbation.  Larger values of $\delta$ produce stronger directional asymmetry.  In the simulation, we vary $\delta \in \{2,3,\ldots,10\}$. The inclusion of coordinate directions in $\mathcal U_m$ ensures that the sparse signal direction is represented in the search set, while the random directions provide additional coverage of the unit sphere.

For each generated sample, we compute $\widehat S_{\mathcal U_m,\mathcal A_S}$ and $\widehat T_{\mathcal U_m}^{+}$ and compare them with their Gaussian Monte Carlo critical values.  The resulting decisions are summarized through the following four diagnostic categories:
\begin{equation*}
\begin{array}{ll}
\text{symmetric benchmark:} & \widehat S_{\mathcal U_m,\mathcal A_S}\leqslant c_S,\quad \widehat T_{\mathcal U_m}^{+}\leqslant c_T,\\
\text{symmetric tail-departed:} & \widehat S_{\mathcal U_m,\mathcal A_S}\leqslant c_S,\quad \widehat T_{\mathcal U_m}^{+}> c_T,\\
\text{skewed benchmark:} & \widehat S_{\mathcal U_m,\mathcal A_S}> c_S,\quad \widehat T_{\mathcal U_m}^{+}\leqslant c_T,\\
\text{skewed tail-departed:} & \widehat S_{\mathcal U_m,\mathcal A_S}> c_S,\quad \widehat T_{\mathcal U_m}^{+}> c_T,
\end{array}
\end{equation*}
where $c_S$ and $c_T$ are the calibrated critical values.  The simulation reports the rejection frequencies of the skewness and tail components separately, together with the proportion of samples assigned to the correct diagnostic category.

\subsection{Simulation results}\label{sec:simulation_results}

We now discuss the finite-sample behavior of the proposed diagnostics.  The main simulation results are reported for the general directional-skewness setting, where the skewness vector is dense: the total magnitude of the skewness vector is $\norm{\boldalpha}=\delta$, but the asymmetric perturbation is spread over all coordinates.  This setting is more challenging than a sparse alternative because no single coordinate direction carries the full skewness signal.  

First, \Cref{fig:skew_results_general} summarizes the rejection probabilities of the skewness component. It shows that the multi-quantile directional skewness diagnostic has good size behavior under symmetric models and increasing power under skewed alternatives.  For the Gaussian, Student-$t$, and contaminated normal models, the rejection probabilities remain close to the nominal level across all values of $n$, $p$, and $\delta$.  This indicates that the skewness component is not responding to symmetric heavy tails or to symmetric contamination.  This feature is essential for the intended interpretation of the diagnostic, since rejection of the skewness component should be evidence against central symmetry rather than evidence of general non-Gaussianity.

\begin{figure}[!ht]
\centering
\includegraphics[width=0.75\textwidth,keepaspectratio]{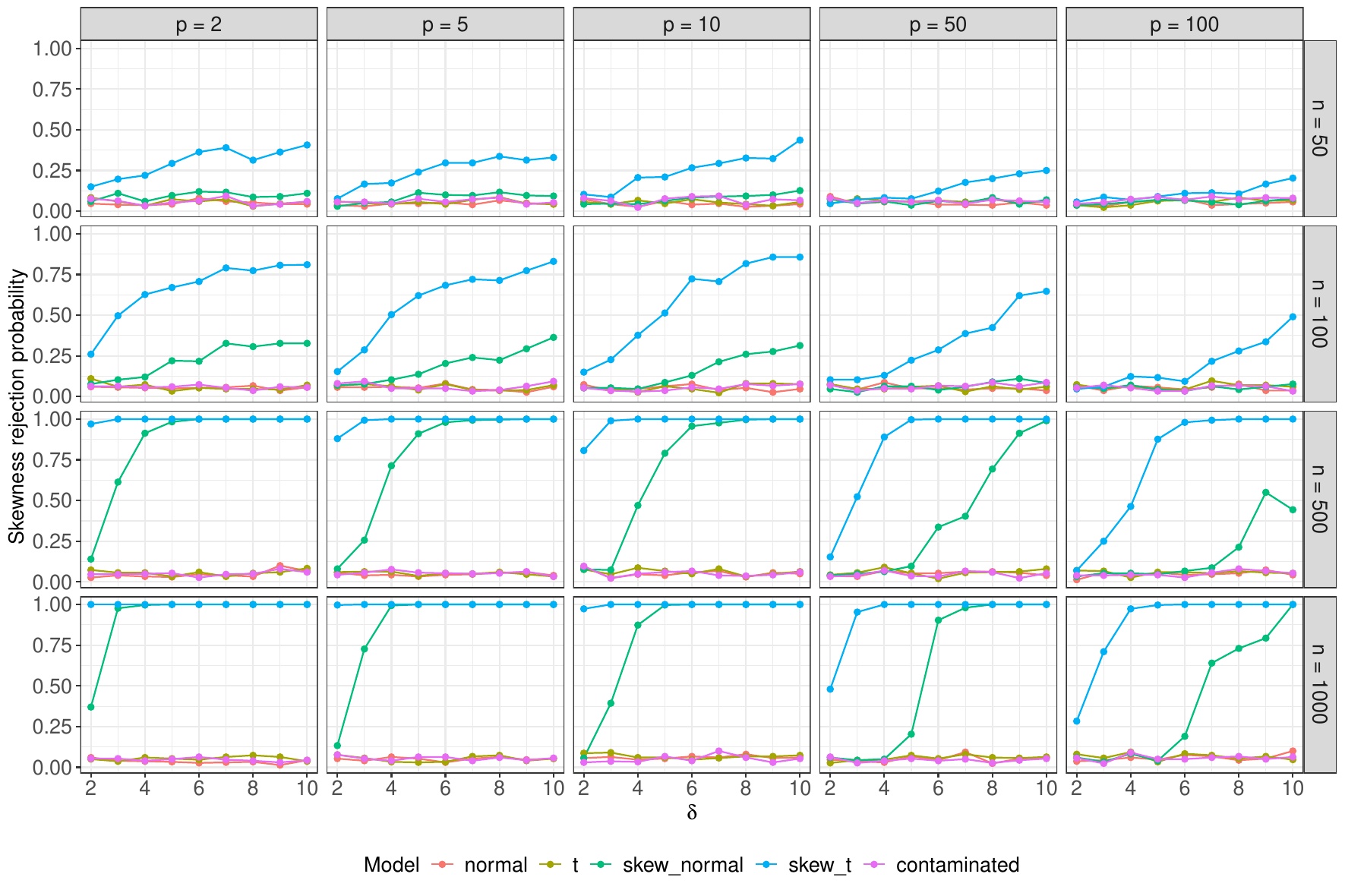}
\caption{Empirical rejection probabilities of the skewness diagnostic under the general directional-skewness setting. The rejection probability is plotted against the skewness-strength parameter $\delta$, with panels corresponding to different values of $n$ and $p$.}
\label{fig:skew_results_general}
\end{figure}

For the skew-normal and skew-$t$ models, the rejection probability increases with the skewness-strength parameter $\delta$.  The trend is especially pronounced for the skew-$t$ model.  Even for moderate sample sizes, the skew-$t$ rejection probability rises substantially with $\delta$, and for $n=500$ or $n=1000$ it is close to one over much of the grid.  This indicates that the skewness component can detect asymmetric heavy-tailed alternatives effectively once the sample size is moderate. On the other hand, the skew-normal model is more challenging, particularly when $n$ is small or $p$ is large.  For $n=50$, the rejection probabilities against the skew-normal alternative are modest across all dimensions, although they tend to increase with $\delta$.  For $n=100$, power improves, but the effect of dimension remains visible.  For $n=500$ and $n=1000$, the power becomes much stronger, especially for moderate and large $\delta$.  In lower and moderate dimensions, the rejection probability approaches one as $\delta$ increases.  In higher dimensions, such as $p=50$ and $p=100$, larger sample sizes and stronger skewness are required.  This behavior is expected in the dense-skewness setting: the skewness signal is distributed across many coordinates, so individual coordinate directions are less informative, and the diagnostic relies more heavily on random directional coverage.

The comparison across dimensions also illustrates the geometric difficulty of detecting directional asymmetry in high dimensions.  Since the skewness direction is dense, random directions must have sufficient alignment with the vector $\alpha$ in order to reveal strong projected asymmetry.  When $p$ is large and $n$ is small, this alignment effect, combined with the finite-sample variability of quantiles, reduces power.  As $n$ increases, the empirical quantile estimates stabilize, and the diagnostic becomes more sensitive to the dense skewness direction.

\begin{figure}[!ht]
\centering
\includegraphics[width=0.8\textwidth,keepaspectratio]{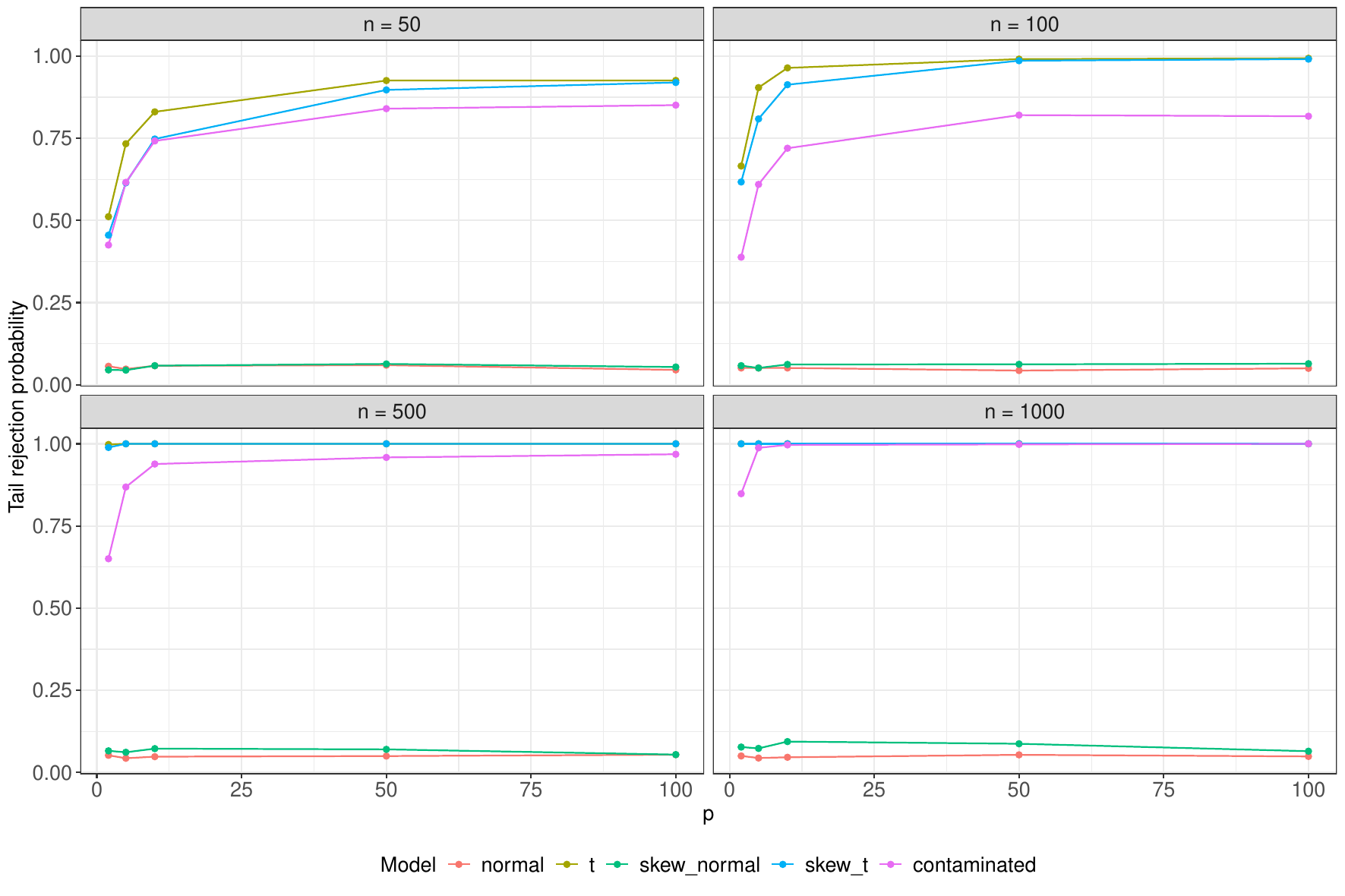}
\caption{Empirical rejection probabilities of the one-sided tail diagnostic under the general directional-skewness setting. Rejection probabilities are plotted against $p$ with panels corresponding to different sample sizes.}
\label{fig:tail_results_general}
\end{figure}

Next, in \Cref{fig:tail_results_general}, we report the rejection probabilities of the one-sided tail diagnostic.  The results show a clear separation between Gaussian tail and heavy-tailed models.  The Gaussian model remains close to the nominal rejection level.  The skew-normal model also remains close to the nominal level, despite being directionally asymmetric.  This confirms that the tail-ratio statistic is not spuriously rejecting merely because the distribution is skewed.  Instead, it responds primarily to tail inflation relative to the Gaussian inter-quantile benchmark. For the multivariate $t$ and skew-$t$ models, the tail rejection probability is high even for small samples and becomes essentially one as the sample size increases.  This is consistent with the scale-mixture construction of these models, which produces heavy-tailed projected distributions in all nondegenerate directions.  The contaminated normal model also shows high rejection probabilities, although the rejection rate is lower for small $n$.  This is natural because only a fraction $\varepsilon=0.05$ of observations comes from the inflated-variance component, and in smaller samples the realized number and extremeness of contaminated observations may vary substantially.  As $n$ increases, the tail diagnostic detects the contamination-driven tail inflation with high probability.

For completeness, we also report the corresponding simulation results for the sparse directional-skewness setting in Figures~\ref{fig:skew_results_sparse} and~\ref{fig:tail_results_sparse}.  Recall that in this setting, $\boldalpha=\delta \bm{e}_1$, so the skewness signal is concentrated in the first coordinate.  Since the direction set includes the coordinate directions, the signal-bearing direction is explicitly present in the projection search.  This makes the sparse setting more favorable for skewness detection than the general setting. The sparse-skewness results show the same qualitative separation between skewness and tail behavior as in the dense case, but the skewness component is more powerful.  This is most apparent for the skew-normal model.  In the dense setting, the skewness signal is spread across all coordinates, and detection in high dimensions requires sufficient random directional alignment with $\boldalpha$.  In the sparse setting, the signal is concentrated in $\bm{e}_1$, which is included among the coordinate directions.  Consequently, the diagnostic can evaluate a projection that is directly aligned with the skewness signal, leading to substantially higher rejection probabilities. The tail diagnostic changes little between the two settings. In both cases, the tail rejection probability remains close to the nominal level for the Gaussian and skew-normal models and is high for the other three models.

\begin{figure}[!ht]
\centering
\includegraphics[width=0.75\textwidth,keepaspectratio]{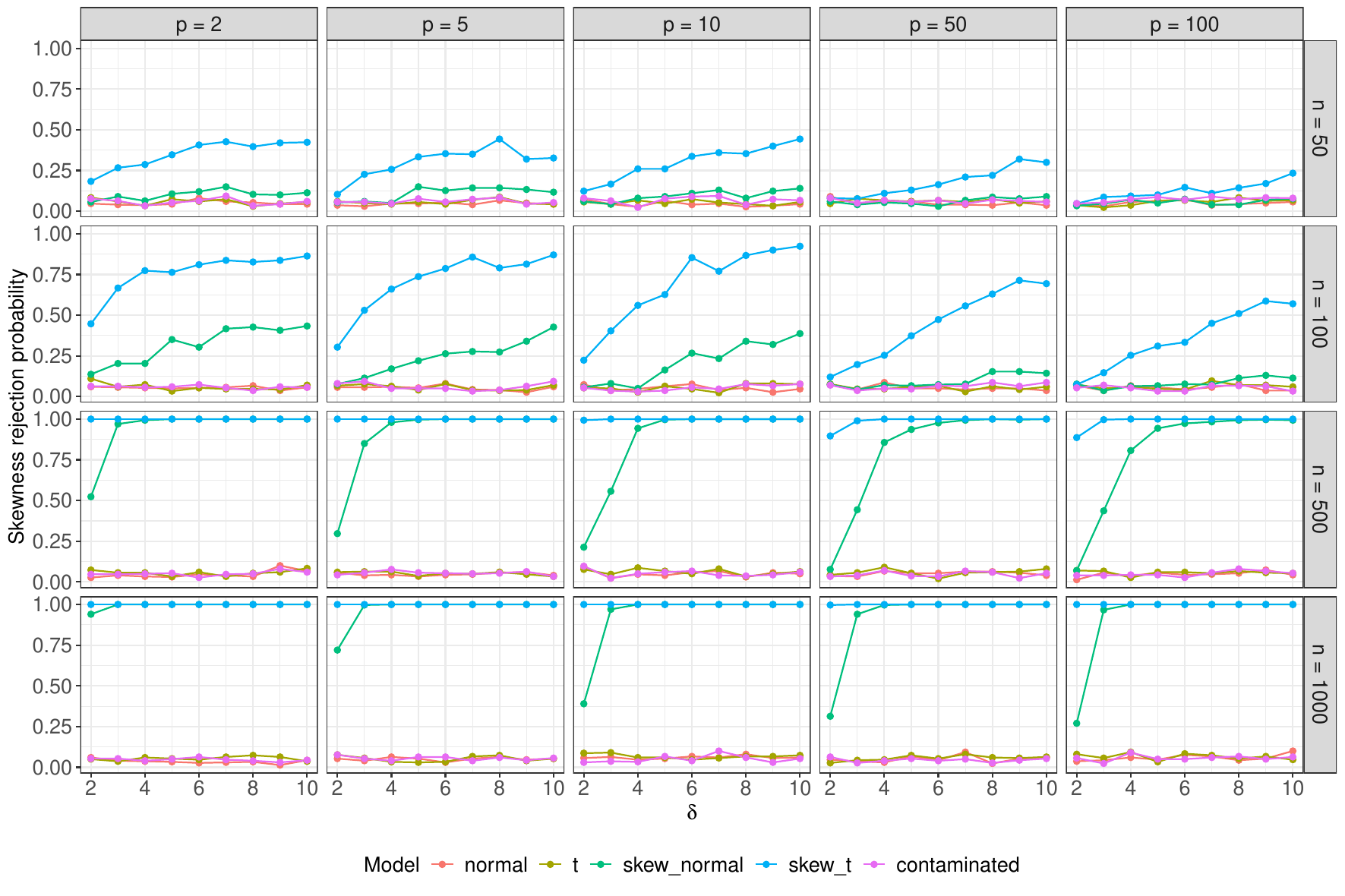}
\caption{Empirical rejection probabilities of the skewness diagnostic under the sparse directional-skewness setting.  Compared with the dense setting in Figure~\ref{fig:skew_results_general}, the skewness rejection probabilities are generally higher, especially for the skew-normal model, because the coordinate direction carrying the skewness signal is included in $\mathcal U_m$.}
\label{fig:skew_results_sparse}
\end{figure}

\begin{figure}[!ht]
\centering
\includegraphics[width=0.8\textwidth,keepaspectratio]{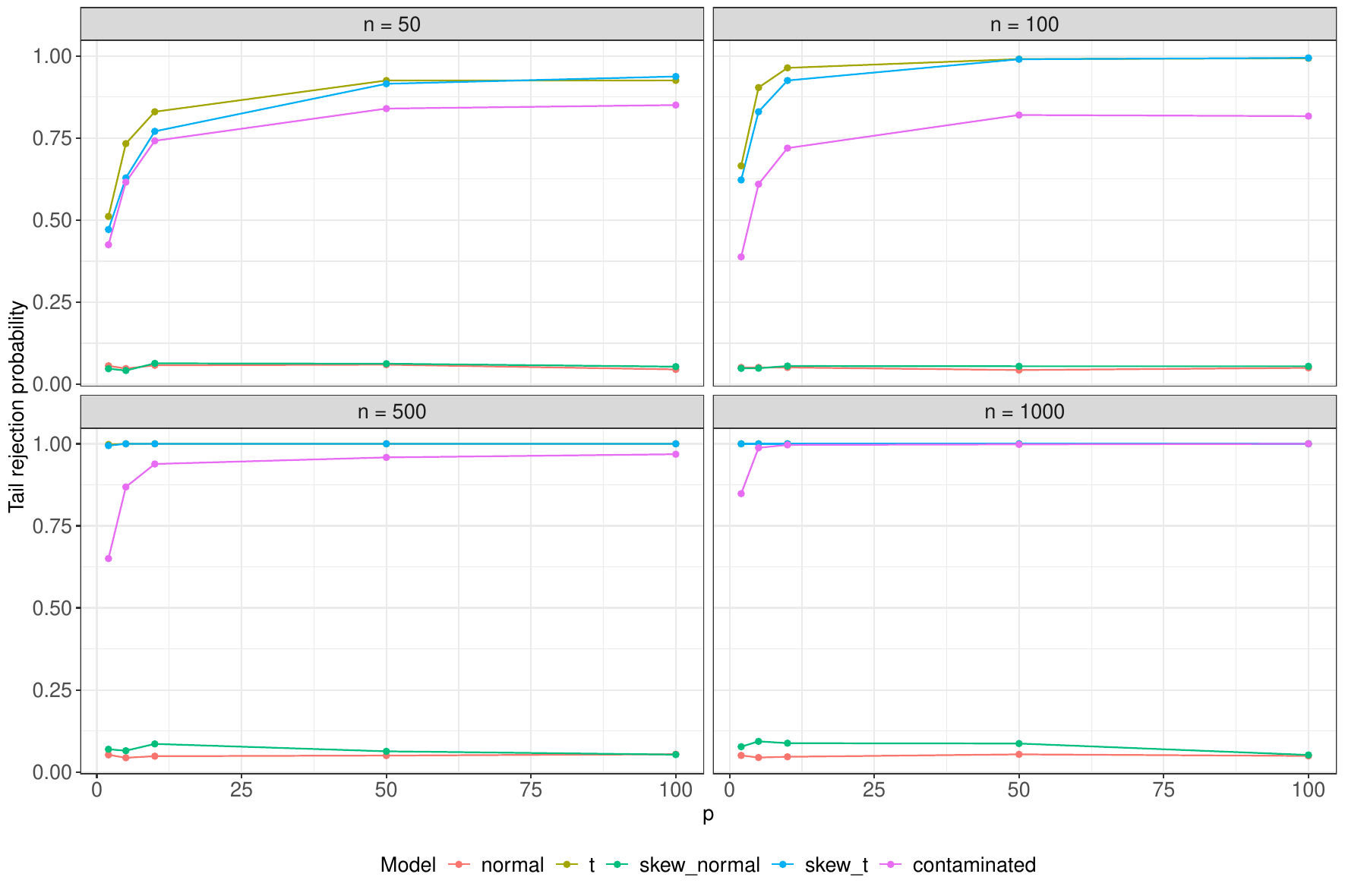}
\caption{Empirical rejection probabilities of the one-sided tail diagnostic under the sparse directional-skewness setting. The qualitative behavior is similar to that in Figure~\ref{fig:tail_results_general}: Student-$t$, skew-$t$, and contaminated normal models are detected as heavy-tailed with high probability, while Gaussian and skew-normal models remain close to the nominal level.}
\label{fig:tail_results_sparse}
\end{figure}

Taken together, these results support the intended separation between the two diagnostic components.  The skewness statistic rejects primarily for skew-normal and skew-$t$ alternatives, while remaining close to the nominal level for symmetric distributions, including symmetric heavy-tailed and contaminated models.  The tail statistic rejects primarily for Student-$t$, skew-$t$, and contaminated normal alternatives, while remaining close to the nominal level for Gaussian and skew-normal data.  The method therefore provides more information than an omnibus test of multivariate normality: it distinguishes evidence of directional asymmetry from evidence of tail inflation. The results also show that tail inflation is generally easier to detect than skewness in the present design.  The tail diagnostic has high power for Student-$t$ and skew-$t$ alternatives even at relatively small sample sizes.  By contrast, skewness detection depends more strongly on $n$, $p$, and $\delta$.  This is particularly visible for the dense skew-normal model, where the asymmetry is distributed across many coordinates and is therefore less visible in any single projection.  Nevertheless, for moderate to large sample sizes, the multi-quantile skewness statistic achieves high power against both skew-normal and skew-$t$ alternatives.

The four-regime interpretation follows directly from these componentwise findings.  Gaussian samples are expected to be classified as symmetric benchmark-tailed; Student-$t$ and contaminated normal samples as symmetric tailed; skew-normal samples as skewed benchmark-tailed; and skew-$t$ samples as skewed tail-departed. Below, in \Cref{tab:classification_dense_sparse}, we present and compare the four-regime classification results for dense and sparse skewness alternatives under specific choices of $n$ and $p$.  The entries are row proportions, so each row describes the empirical distribution of the selected diagnostic category conditional on the true data-generating mechanism.  The table therefore shows not only whether the classifier is correct, but also the type of error made when misclassification occurs.

\begin{table}[!ht]
\centering
\scriptsize
\caption{Four-regime classification proportions for the general and sparse skewness settings with $\delta=5$. Entries are row proportions over Monte Carlo replications. The columns SB, KB, ST, and KT denote symmetric benchmark, skewed benchmark, symmetric tail-departed, and skewed tail-departed, respectively.}
\label{tab:classification_dense_sparse}
\begin{tabular}{llcccc@{\hspace{2em}}cccc}
\toprule
& & \multicolumn{4}{c}{General skewness structure} & \multicolumn{4}{c}{Sparse skewness structure} \\
\cmidrule(lr){3-6} \cmidrule(lr){7-10}
Setting & True model
& SB & KB & ST & KT
& SB & KB & ST & KT \\
\midrule
$n=100$ $p=10$& Normal              & 0.883 & 0.057 & 0.053 & 0.007 & 0.883 & 0.057 & 0.053 & 0.007 \\
& Student-$t$         & 0.033 & 0.003 & 0.903 & 0.060 & 0.033 & 0.003 & 0.903 & 0.060 \\
& Skew-normal         & 0.850 & 0.077 & 0.063 & 0.010 & 0.777 & 0.147 & 0.060 & 0.017 \\
& Skew-$t$            & 0.067 & 0.030 & 0.420 & 0.483 & 0.050 & 0.043 & 0.323 & 0.583 \\
& Contaminated & 0.290 & 0.000 & 0.673 & 0.037 & 0.290 & 0.000 & 0.673 & 0.037 \\
\midrule
$n=500,\ p=10$ & Normal              & 0.920 & 0.040 & 0.040 & 0.000 & 0.920 & 0.040 & 0.040 & 0.000 \\
& Student-$t$         & 0.000 & 0.000 & 0.933 & 0.067 & 0.000 & 0.000 & 0.933 & 0.067 \\
& Skew-normal         & 0.193 & 0.740 & 0.017 & 0.050 & 0.003 & 0.917 & 0.000 & 0.080 \\
& Skew-$t$            & 0.000 & 0.000 & 0.000 & 1.000 & 0.000 & 0.000 & 0.000 & 1.000 \\
& Contaminated & 0.090 & 0.007 & 0.850 & 0.053 & 0.090 & 0.007 & 0.850 & 0.053 \\
\midrule
$n=500,\ p=50$ & Normal              & 0.900 & 0.053 & 0.047 & 0.000 & 0.900 & 0.053 & 0.047 & 0.000 \\
& Student-$t$         & 0.000 & 0.000 & 0.947 & 0.053 & 0.000 & 0.000 & 0.947 & 0.053 \\
& Skew-normal         & 0.853 & 0.093 & 0.050 & 0.003 & 0.063 & 0.890 & 0.000 & 0.047 \\
& Skew-$t$            & 0.000 & 0.000 & 0.003 & 0.997 & 0.000 & 0.000 & 0.000 & 1.000 \\
& Contaminated & 0.040 & 0.003 & 0.923 & 0.033 & 0.040 & 0.003 & 0.923 & 0.033 \\
\midrule
$n=1000,\ p=100$ & Normal              & 0.940 & 0.043 & 0.013 & 0.003 & 0.940 & 0.043 & 0.013 & 0.003 \\
& Student-$t$         & 0.000 & 0.000 & 0.967 & 0.033 & 0.000 & 0.000 & 0.967 & 0.033 \\
& Skew-normal         & 0.923 & 0.037 & 0.037 & 0.003 & 0.000 & 0.967 & 0.000 & 0.033 \\
& Skew-$t$            & 0.000 & 0.000 & 0.003 & 0.997 & 0.000 & 0.000 & 0.000 & 1.000 \\
& Contaminated & 0.000 & 0.000 & 0.950 & 0.050 & 0.000 & 0.000 & 0.950 & 0.050 \\
\bottomrule
\end{tabular}
\end{table}

The table illustrates that for the Gaussian model, most samples are assigned to the symmetric benchmark category, with rejection rates close to the nominal level. For the Student-$t$ and contaminated normal models, most samples are assigned to the symmetric tailed category.  Thus, the classifier continues to separate symmetric heavy-tailed behavior from directional skewness. The main difference between the two settings appears for the skew-normal model.  In the general case, the skew-normal alternative is considerably harder to classify correctly in high dimensions.  For example, at $n=500$ and $p=50$, only $0.093$ of dense skew-normal samples are classified as skewed benchmark, while $0.853$ are classified as symmetric benchmark.  At $n=1000$ and $p=100$, the same pattern persists.  This reflects the difficulty of detecting a dense but weakly projected skewness direction when the signal is spread over many coordinates. In contrast, the sparse skew-normal alternative is classified much more accurately once the sample size is moderate.  At $n=500$ and $p=10$, the sparse skew-normal model is classified as skewed benchmark with probability $0.917$, compared with $0.740$ in the dense case. At $n=500$ and $p=50$, the difference is more pronounced: the sparse case has a correct classification probability of $0.890$, whereas the dense case has only $0.093$.  At $n=1000$ and $p=100$ as well, the results are substantially better in the sparse case.  This comparison highlights the role of the direction set. Under the sparse skewness setting, the signal-bearing projection is directly searched.  For dense skewness, the signal direction is not necessarily well represented by the finite set of random and coordinate directions, especially in high dimensions. Moving on to the skew-$t$ model, we notice that classification is strong in both settings once the sample size is moderate. For $n=500$ and larger, both dense and sparse skew-$t$ samples are classified almost always as skewed tail-departed.  The heavy-tailed scale-mixture component helps the diagnostic detect departure from the Gaussian benchmark, while the skewness component becomes sufficiently strong at these sample sizes.

Overall, the table confirms the main qualitative message of the simulation study.  The tail component is stable across dense and sparse skewness designs and correctly identifies Student-$t$, skew-$t$, and contaminated normal models as tailed.  The skewness component is sensitive to how well the skewness direction is represented in the projection set.  Sparse skewness is easier to detect because the coordinate directions include the signal direction.  Dense skewness can be substantially harder in high dimensions, suggesting that additional structured directions, such as PCA or ICA directions, may be useful in settings where the asymmetric direction is not sparse.

\subsection{Comparison against existing methods}

We also compared the proposed diagnostics with several standard procedures for assessing multivariate normality, including Mardia's skewness and kurtosis tests \citep{mardia1970measures}, the Henze--Zirkler test \citep{henze1990class}, the Doornik--Hansen test \citep{doornik2008omnibus}, the energy test \citep{szekely2013energy}, and the Malkovich--Afifi projection-kurtosis test \citep{malkovich1973tests}. These tests are implemented using the \texttt{R} packages \texttt{MVN} \citep{MVNpackage} and \texttt{mnt} \citep{mntpackage}. The broad conclusion from these comparisons is that the proposed method has rejection behavior comparable to existing methods under the main non-Gaussian alternatives considered in the simulation study.  In particular, the proposed skewness component has competitive power against skewed alternatives when the directional signal is sufficiently represented in the projection set, and the proposed tail component has high power against Student-\(t\), skew-\(t\), and contaminated normal alternatives. We also note that the Malkovich--Afifi projection-kurtosis test is extremely time-consuming and therefore, albeit projection based, should not be preferred over our proposed technique.

The comparison also highlights an important distinction between the proposed procedure and existing omnibus tests.  Methods such as the Henze--Zirkler, Doornik--Hansen, and energy tests are effective at detecting departures from multivariate normality, but their rejection decisions do not identify the nature of the departure.  Similarly, moment-based tests such as Mardia's skewness and kurtosis tests provide useful classical benchmarks, but they can be sensitive to moment existence, tail behavior, and high-dimensional finite-sample effects.  In our simulations, the proposed quantile-based diagnostics gave broadly similar power while retaining a more interpretable two-component structure: one statistic targets directional asymmetry, and the other targets tail inflation relative to a benchmark.

Thus, the main contribution of the proposed method is not that it uniformly dominates all existing tests in raw rejection probability.  Rather, it provides a robust and interpretable diagnostic decomposition of multivariate non-Gaussianity.  An omnibus test may conclude that a sample is non-normal, but it does not indicate whether a symmetric heavy-tailed model, a skewed benchmark-tailed model, or a skewed heavy-tailed model is more appropriate.  The proposed procedure is designed precisely for this screening task, and the simulation results suggest that this additional interpretability is obtained without a substantial loss of power relative to standard competing methods.

\section{Application to Indian air-quality data}\label{sec:real_data}

As a real-life application, we illustrate the proposed diagnostic on Indian air-quality data for five major cities: Delhi, Mumbai, Chennai, Kolkata, and Bangalore. The data have been originally sourced from the Central Pollution Control Board of the Government of India\footnote{CPCB: \url{https://cpcb.nic.in/}}, and a pre-processed version is publicly available in Kaggle\footnote{Kaggle data source: \url{https://www.kaggle.com/datasets/ankushpanday1/air-quality-data-in-india-2015-2024/}}. The objective of this analysis is not to build a forecasting model for pollutant concentrations, but to demonstrate how the proposed procedure can be used as a preliminary distributional screening tool for multivariate environmental data.  Such screening is useful because pollutant vectors often contain episodic spikes, asymmetric responses to meteorological conditions, and joint tail events, all of which may invalidate simple Gaussian modeling assumptions. For each city, we construct a 10-dimensional pollutant vector using
\begin{equation}
        \X_t =
        \left(\mathrm{PM}_{2.5,t},\mathrm{PM}_{10,t},\mathrm{NO}_{2,t},\mathrm{NH}_{3,t},\mathrm{SO}_{2,t},\mathrm{CO}_{t},\mathrm{O}_{3,t},\mathrm{Benzene}_{t}, \mathrm{Toluene}_{t},\mathrm{Xylene}_{t}\right)^\top,
\end{equation}
where $t$ stands for different days. To reduce short-range temporal dependence and to obtain comparable sample sizes across cities, we follow the idea of thinning and retain only observations recorded on days spread out from each other. This produces 522 weekly observations for each city over the period of 2015 to 2024. Note that the data vector combines particulate matter, gaseous precursors, combustion-related pollutants, ozone, and volatile organic compounds.  We avoid using \(\mathrm{NO}\), \(\mathrm{NO}_2\), and \(\mathrm{NO_x}\) simultaneously, since that would introduce unnecessary redundancy and high correlation. The variables benzene, toluene, and xylene represent a standard group of volatile organic compounds and provide information about traffic and combustion-related pollution.

For implementation, we first consider the entire 10 year period; and then split the data into two equal subperiods: before 2020 (pre-COVID) and 2020 onward (post-COVID). This split is intended to examine whether the distributional classification changes over time. A similar analysis for Indian cities was conducted by \cite{dhar2025comparative}. Also, before applying the diagnostics, each pollutant variable is transformed using \(\log(1+x)\) and then standardized componentwise. The projection diagnostics are computed using the same tuning choices as in the simulation study.  We use the quantile grid $\mathcal A_S=\{0.05,0.10,0.15,0.20,0.25\}$ for the directional skewness statistic and \(q=0.025\) for the one-sided tail statistic.  The projection set consists of random directions and coordinate directions.  Critical values are obtained by Gaussian Monte Carlo calibration using the same sample size, dimension, projection set, and quantile grid as in the observed data.  The resulting decisions are reported in Table~\ref{tab:air_quality_results}.

\begin{table}[H]
\centering
\scriptsize
\caption{Projection diagnostics for Indian air-quality pollutant vectors.  The columns \(S\) and \(T\) report the observed skewness and one-sided tail statistics, while \(c_S\) and \(c_T\) are their Gaussian Monte Carlo critical values.  The full period uses the complete sample, while the two subperiods split the data into pre-2020 and 2020 onward periods.}
\label{tab:air_quality_results}
\begin{tabular}{llccccc}
\toprule
Period & City & \(S\) & \(c_S\) & \(T\) & \(c_T\) & Diagnostic class \\
\midrule
Full sample & Delhi     & 0.593 & 0.231 & 0.893 & 0.739 & Skewed tail-departed \\
Full sample & Mumbai    & 0.608 & 0.227 & 1.014 & 0.747 & Skewed tail-departed \\
Full sample & Chennai   & 0.607 & 0.227 & 1.277 & 0.728 & Skewed tail-departed \\
Full sample & Kolkata   & 0.585 & 0.232 & 0.849 & 0.737 & Skewed tail-departed \\
Full sample & Bangalore & 0.633 & 0.224 & 0.977 & 0.740 & Skewed tail-departed \\
\midrule
Pre-2020 & Delhi     & 0.671 & 0.320 & 1.289 & 1.150 & Skewed tail-departed \\
Pre-2020 & Mumbai    & 0.647 & 0.321 & 0.980 & 1.120 & Skewed benchmark \\
Pre-2020 & Chennai   & 0.631 & 0.322 & 1.472 & 1.130 & Skewed tail-departed \\
Pre-2020 & Kolkata   & 0.597 & 0.321 & 0.947 & 1.110 & Skewed benchmark \\
Pre-2020 & Bangalore & 0.630 & 0.322 & 1.129 & 1.140 & Skewed benchmark \\
\midrule
2020 onward & Delhi     & 0.628 & 0.324 & 1.590 & 1.130 & Skewed tail-departed \\
2020 onward & Mumbai    & 0.607 & 0.321 & 1.370 & 1.120 & Skewed tail-departed \\
2020 onward & Chennai   & 0.634 & 0.318 & 1.200 & 1.120 & Skewed tail-departed \\
2020 onward & Kolkata   & 0.645 & 0.317 & 1.150 & 1.130 & Skewed tail-departed \\
2020 onward & Bangalore & 0.652 & 0.321 & 1.370 & 1.140 & Skewed tail-departed \\
\bottomrule
\end{tabular}
\end{table}

The full-sample results show a strikingly consistent pattern.  For all five cities, both statistics exceed their corresponding Gaussian critical values.  Consequently, all five cities are classified as skewed tail-departed.  This suggests that the joint distribution of the pollutant vector is not adequately described by a Gaussian or symmetric benchmark-tailed model.  It also suggests that a symmetric heavy-tailed model alone would be incomplete, since the skewness component rejects strongly.  The diagnostic therefore points toward the use of skewed heavy-tailed multivariate models, such as skew-\(t\)-type models, as more appropriate distributional working models for the full-period pollutant vectors. On this note, we refer to the work by \cite{bartoletti2010modelling} who argued that multivariate skewed models are more suitable for air quality data analysis.

The subperiod analysis reveals a more nuanced temporal pattern.  In the pre-2020 period, all five cities show strong evidence of directional skewness.  However, the evidence for tail inflation is city-dependent.  Delhi and Chennai are classified as skewed tail-departed, whereas Mumbai, Kolkata, and Bangalore are classified as skewed benchmark.  Thus, before 2020, the multivariate pollutant distributions for Mumbai, Kolkata, and Bangalore exhibit directional asymmetry, but their tail statistics do not exceed the Gaussian tail critical values at the chosen level.  For these cities and this period, a skewed but not necessarily heavy-tailed model is suggested by the diagnostic. One may in fact introduce other types of transformation to reduce the extent of skewness and model these under a multivariate Gaussian framework. Contrary to this, in the post-2020 period, all five cities are classified as skewed tail-departed. The most pronounced tail-ratio departure occur for Delhi, followed by Mumbai and Bangalore. From a modeling perspective, this suggests that skewed heavy-tailed distributions are more appropriate for the post-2020 pollutant vectors across all five cities.

These results illustrate the value of separating the two components of non-Gaussianity.  A standard omnibus normality test would indicate only that the pollutant vectors are non-Gaussian.  The proposed procedure provides a more interpretable diagnosis.  For the full sample and the post-2020 period, the departure is both asymmetric and heavy-tailed.  For several cities in the pre-2020 period, the departure is primarily asymmetric rather than strongly tail-inflated.  This distinction is useful because it suggests different modeling choices: skewed benchmark-tailed classifications motivate asymmetric models without necessarily requiring heavy-tailed radial components, while skewed tail-departed classifications motivate asymmetric heavy-tailed models. On this note, we also find it pertinent to highlight the time ordering of the observations. Although we follow the idea of thinning to reduce short-range serial dependence, the observations are still temporally ordered and may exhibit seasonality or longer-range dependence. Therefore, the present analysis should be interpreted as a marginal distributional diagnostic rather than as a formal time-series test with exact size control. A dependence-adjusted calibration, for example using seasonal or block resampling, would be a natural extension for formal time-series inference.  Nevertheless, the magnitude of the observed statistics relative to their Gaussian critical values indicates substantial departures from symmetric benchmark-tailed behavior in the empirical pollutant distributions.

\section{Conclusions}\label{sec:conclusions}

We have developed a projection-based framework for diagnosing directional asymmetry and tail-ratio departure in multivariate data. The procedure works with one-dimensional projections, combining random directions with coordinate directions, and uses quantile-based summaries to reduce dependence on moment assumptions. The theoretical results show that the empirical diagnostics are uniformly close to their population targets over the searched directions and skewness levels, and that the resulting four-regime classification is consistent under separation. The sparse-direction calculation also clarifies why coordinate augmentation can be helpful when the departure is concentrated in a few variables. The framework is intended as an interpretable diagnostic tool that can guide subsequent modeling choices, especially when deciding whether a symmetric, skewed, tail-departed, or combined model is appropriate.

The numerical results support this interpretation. Across the simulation designs, the skewness component remains close to the nominal level under symmetric models, while gaining power against skew-normal and skew-\(t\) alternatives as the skewness strength and sample size increase. The tail-ratio component shows strong power against all heavy-tailed alternatives, while remaining stable under Gaussian and skew-normal models. Taken together, these findings indicate that the two components are not simply duplicating an omnibus test of multivariate normality. Rather, they provide separate evidence about two different forms of non-Gaussianity: violation of central symmetry and tail-ratio departure from a benchmark law. The simulations also clarify the role of the direction set. Dense directional skewness is intrinsically harder to detect in high dimensions because the signal is spread over many coordinates and a finite random search may not align closely with the skewness direction. Sparse skewness is easier to detect when coordinate directions are included, because the signal-bearing coordinate is explicitly searched. This agrees with the rank-one calculations in the theory and suggests a practical interpretation of the method: random directions provide broad coverage, while coordinate directions protect against marginal or sparse departures. The comparison with existing tests further shows that the proposed method has broadly comparable rejection behavior under the main alternatives, while offering a more interpretable decomposition than standard omnibus normality tests. Finally, with the real-data analysis for the Indian air-quality data, the proposed diagnostics identify both directional asymmetry and tail-ratio departure in the multivariate pollutant vectors, with some differences across cities and time periods. Such information is useful before fitting a multivariate distributional model. A standard omnibus test may indicate that the pollutant vector is non-Gaussian, but it does not say whether a symmetric heavy-tailed model, a skewed benchmark-tailed model, or a skewed heavy-tailed model is more appropriate. The proposed diagnostic provides exactly this type of preliminary model guidance.

Several extensions remain open. First, although the present implementation uses random and coordinate directions, the same framework can incorporate data-adaptive directions, including PCA, ICA, or projection-pursuit directions. A careful theory for such directions would require accounting for their data dependence, possibly through sample splitting or perturbation arguments. Second, the calibration step can be developed further. In this paper, Gaussian calibration is used as the main reference, while elliptical calibration is discussed as a possible alternative. More detailed theory for bootstrap and semi-parametric calibration would be useful, especially when the reference model is not Gaussian. Third, the present theory is developed for independent observations. Some applications naturally involve temporal dependence and seasonality, so block-resampling or dependence-adjusted versions of the diagnostic would be valuable for formal time-series inference. Finally, extensions to functional, compositional, or ultra high-dimensional data would broaden the scope of the method and connect it more directly to modern multivariate applications.

\bibliographystyle{cas-model2-names}
\bibliography{ref.bib}

@article{dvoretzky1956asymptotic,
	title={Asymptotic minimax character of the sample distribution function and of the classical multinomial estimator},
	author={Dvoretzky, Aryeh and Kiefer, Jack and Wolfowitz, Jacob},
	journal={The Annals of Mathematical Statistics},
	volume={27},
	number={3},
	pages={642--669},
	year={1956}
}

@article{massart1990tight,
	title={The tight constant in the {Dvoretzky--Kiefer--Wolfowitz} inequality},
	author={Massart, Pascal},
	journal={The Annals of Probability},
	volume={18},
	number={3},
	pages={1269--1283},
	year={1990}
}

@book{vandervaart1998asymptotic,
	title={Asymptotic Statistics},
	author={van der Vaart, Aad W.},
	publisher={Cambridge University Press},
	year={1998}
}

@book{ledoux2001concentration,
	title={The Concentration of Measure Phenomenon},
	author={Ledoux, Michel},
	publisher={American Mathematical Society},
	year={2001}
}

@book{vershynin2018highdimensional,
	title={High-Dimensional Probability: An Introduction with Applications in Data Science},
	author={Vershynin, Roman},
	edition={2nd},
	publisher={Cambridge University Press},
	year={2026}
}

@article{cambanis1981elliptically,
	title={On the theory of elliptically contoured distributions},
	author={Cambanis, Stamatis and Huang, Steel and Simons, Gordon},
	journal={Journal of Multivariate Analysis},
	volume={11},
	number={3},
	pages={368--385},
	year={1981}
}

@article{malkovich1973tests,
	title   = {On tests for multivariate normality},
	author  = {Malkovich, James F. and Afifi, A. A.},
	journal = {Journal of the American Statistical Association},
	volume  = {68},
	number  = {341},
	pages   = {176--179},
	year    = {1973}
}

@article{baringhaus1991limit,
	title   = {Limit distributions for measures of multivariate skewness and kurtosis based on projections},
	author  = {Baringhaus, Ludwig and Henze, Norbert},
	journal = {Journal of Multivariate Analysis},
	volume  = {38},
	number  = {1},
	pages   = {51--69},
	year    = {1991}
}

@article{kong2012quantile,
	title   = {Quantile tomography: Using quantiles with multivariate data},
	author  = {Kong, Linglong and Mizera, Ivan},
	journal = {Statistica Sinica},
	volume  = {22},
	number  = {4},
	pages   = {1589--1610},
	year    = {2012}
}

@article{crow1967robust,
	title   = {Robust estimation of location},
	author  = {Crow, Edwin L. and Siddiqui, M. M.},
	journal = {Journal of the American Statistical Association},
	volume  = {62},
	number  = {318},
	pages   = {353--389},
	year    = {1967}
}

@article{ruppert1987what,
	title   = {What is kurtosis? An influence function approach},
	author  = {Ruppert, David},
	journal = {The American Statistician},
	volume  = {41},
	number  = {1},
	pages   = {1--5},
	year    = {1987}
}

@article{jones2011skewness,
	title   = {Skewness-invariant measures of kurtosis},
	author  = {Jones, M. C. and Rosco, J. F. and Pewsey, Arthur},
	journal = {The American Statistician},
	volume  = {65},
	number  = {2},
	pages   = {89--95},
	year    = {2011}
}

@book{fang1990symmetric,
	title={Symmetric Multivariate and Related Distributions},
	author={Fang, Kai-Tai and Kotz, Samuel and Ng, Kai Wang},
	publisher={Chapman and Hall},
	year={1990}
}

@article{hall1993almost,
	title   = {On almost linearity of low dimensional projections from high dimensional data},
	author  = {Hall, Peter and Li, Ker-Chau},
	journal = {The Annals of Statistics},
	volume  = {21},
	number  = {2},
	pages   = {867--889},
	year    = {1993}
}

@article{ledoit2004well,
	title   = {A well-conditioned estimator for large-dimensional covariance matrices},
	author  = {Ledoit, Olivier and Wolf, Michael},
	journal = {Journal of Multivariate Analysis},
	volume  = {88},
	number  = {2},
	pages   = {365--411},
	year    = {2004}
}

@article{chen2010shrinkage,
	title   = {Shrinkage algorithms for {MMSE} covariance estimation},
	author  = {Chen, Yilun and Wiesel, Ami and Eldar, Yonina C. and Hero, Alfred O.},
	journal = {IEEE Transactions on Signal Processing},
	volume  = {58},
	number  = {10},
	pages   = {5016--5029},
	year    = {2010}
}

@article{bickel2008covariance,
	title   = {Covariance regularization by thresholding},
	author  = {Bickel, Peter J. and Levina, Elizaveta},
	journal = {The Annals of Statistics},
	volume  = {36},
	number  = {6},
	pages   = {2577--2604},
	year    = {2008}
}

@article{cai2011adaptive,
	title   = {Adaptive thresholding for sparse covariance matrix estimation},
	author  = {Cai, Tony T. and Liu, Weidong},
	journal = {Journal of the American Statistical Association},
	volume  = {106},
	number  = {494},
	pages   = {672--684},
	year    = {2011}
}

@article{maronna1976robust,
	title   = {Robust {M}-estimators of multivariate location and scatter},
	author  = {Maronna, Ricardo A.},
	journal = {The Annals of Statistics},
	volume  = {4},
	number  = {1},
	pages   = {51--67},
	year    = {1976}
}

@article{tyler1987distribution,
	title   = {A distribution-free {M}-estimator of multivariate scatter},
	author  = {Tyler, David E.},
	journal = {The Annals of Statistics},
	volume  = {15},
	number  = {1},
	pages   = {234--251},
	year    = {1987}
}

@article{friedman1974projection,
	title   = {A projection pursuit algorithm for exploratory data analysis},
	author  = {Friedman, Jerome H. and Tukey, John W.},
	journal = {IEEE Transactions on Computers},
	volume  = {C-23},
	number  = {9},
	pages   = {881--890},
	year    = {1974}
}

@article{diaconis1984asymptotics,
	title   = {Asymptotics of graphical projection pursuit},
	author  = {Diaconis, Persi and Freedman, David},
	journal = {The Annals of Statistics},
	volume  = {12},
	number  = {3},
	pages   = {793--815},
	year    = {1984}
}

@article{pena2001multivariate,
	title   = {Multivariate outlier detection and robust covariance matrix estimation},
	author  = {Pe{\~n}a, Daniel and Prieto, Francisco J.},
	journal = {Technometrics},
	volume  = {43},
	number  = {3},
	pages   = {286--300},
	year    = {2001}
}

@book{hyvarinen2001ica,
	title     = {Independent Component Analysis},
	author    = {Hyv{\"a}rinen, Aapo and Karhunen, Juha and Oja, Erkki},
	publisher = {John Wiley \& Sons},
	year      = {2001}
}

@article{loperfido2018skewness,
	title   = {Skewness-based projection pursuit: A computational approach},
	author  = {Loperfido, Nicola},
	journal = {Computational Statistics \& Data Analysis},
	volume  = {120},
	pages   = {42--57},
	year    = {2018}
}

@article{brys2004robust,
	title   = {A Robust Measure of Skewness},
	author  = {Brys, Guy and Hubert, Mia and Struyf, Anja},
	journal = {Journal of Computational and Graphical Statistics},
	volume  = {13},
	number  = {4},
	pages   = {996--1017},
	year    = {2004}
}

@article{chaudhuri1996geometric,
  title={On a geometric notion of quantiles for multivariate data},
  author={Chaudhuri, Probal},
  journal={Journal of the American statistical association},
  volume={91},
  number={434},
  pages={862--872},
  year={1996},
  publisher={Taylor \& Francis}
}

@article{mardia1970measures,
  title={Measures of multivariate skewness and kurtosis with applications},
  author={Mardia, Kanti V},
  journal={Biometrika},
  volume={57},
  number={3},
  pages={519--530},
  year={1970},
  publisher={Oxford University Press}
}

@article{mardia1974applications,
  title={Applications of some measures of multivariate skewness and kurtosis in testing normality and robustness studies},
  author={Mardia, Kanti V},
  journal={Sankhy{\=a}: The Indian Journal of Statistics, Series B},
  pages={115--128},
  year={1974},
  publisher={JSTOR}
}

@article{henze1990class,
  title={A class of invariant consistent tests for multivariate normality},
  author={Henze, Norbert and Zirkler, Bernd},
  journal={Communications in statistics-Theory and Methods},
  volume={19},
  number={10},
  pages={3595--3617},
  year={1990},
  publisher={Taylor \& Francis}
}

@article{doornik2008omnibus,
  title={An omnibus test for univariate and multivariate normality},
  author={Doornik, Jurgen A and Hansen, Henrik},
  journal={Oxford bulletin of economics and statistics},
  volume={70},
  pages={927--939},
  year={2008},
  publisher={Wiley Online Library}
}

@article{szekely2017energy,
  title={The energy of data},
  author={Sz{\'e}kely, G{\'a}bor J and Rizzo, Maria L},
  journal={Annual Review of Statistics and Its Application},
  volume={4},
  pages={447--479},
  year={2017},
  publisher={Annual Reviews}
}

@article{szekely2013energy,
  title={Energy statistics: A class of statistics based on distances},
  author={Sz{\'e}kely, G{\'a}bor J and Rizzo, Maria L},
  journal={Journal of statistical planning and inference},
  volume={143},
  number={8},
  pages={1249--1272},
  year={2013},
  publisher={Elsevier}
}

@article{chowdhury2022sub,
  title={Sub-dimensional Mardia measures of multivariate skewness and kurtosis},
  author={Chowdhury, Joydeep and Dutta, Subhajit and Arellano-Valle, Reinaldo B and Genton, Marc G},
  journal={Journal of Multivariate Analysis},
  volume={192},
  pages={105089},
  year={2022},
  publisher={Elsevier}
}

@Article{MVNpackage,
    title = {MVN: An R Package for Assessing Multivariate Normality.},
    author = {Selcuk Korkmaz and Dincer Goksuluk and Gokmen Zararsiz},
    journal = {The R Journal},
    year = {2014},
    volume = {6},
    number = {2},
    pages = {151--162},
    url = {https://journal.r-project.org/articles/RJ-2014-031/RJ-2014-031.pdf},
}

@Manual{mntpackage,
    title = {mnt: Affine Invariant Tests of Multivariate Normality},
    author = {Lucas Butsch and Bruno Ebner},
    year = {2020},
    note = {R package version 1.3},
    url = {https://CRAN.R-project.org/package=mnt},
    doi = {10.32614/CRAN.package.mnt},
}

@article{bartoletti2010modelling,
  title={Modelling air pollution data by the skew-normal distribution},
  author={Bartoletti, Silvia and Loperfido, Nicola},
  journal={Stochastic Environmental Research and Risk Assessment},
  volume={24},
  number={4},
  pages={513--517},
  year={2010},
  publisher={Springer}
}

@article{dhar2025comparative,
  title={A comparative study of air quality between pre and post COVID-19 periods in India},
  author={Dhar, Sarbendu Bikash},
  journal={Environment, Development and Sustainability},
  volume={27},
  number={1},
  pages={1829--1853},
  year={2025},
  publisher={Springer}
}

\end{document}